\documentclass{article}

\usepackage{arxiv}
\usepackage{comment}
\usepackage{amsmath}
\usepackage{hyperref}
\usepackage{amsthm}
\usepackage{tikz-cd}

\usepackage[utf8]{inputenc} 
\usepackage[T1]{fontenc}    
\usepackage{hyperref}       
\usepackage{url}            
\usepackage{booktabs}       
\usepackage{amsfonts}       
\usepackage{nicefrac}       
\usepackage{microtype}      
\usepackage{lipsum}		
\usepackage{graphicx}
\usepackage{subfig}
\usepackage{doi}
\usepackage{xcolor}
\usepackage[font=scriptsize,labelfont=bf]{caption}
\usepackage{textgreek}

\newtheorem{theorem}{\bf Theorem}[section]
\newtheorem{lemma}{\bf Lemma}[section]
\newtheorem{proposition}{\bf Proposition}[section]

\title{Riemannian Geometry and Molecular Surfaces I: Spectrum of the Laplacian}

\author{Rachael Pirie \\
    School of Natural and Environmental Sciences\\
    Newcastle University\\
	\texttt{r.pirie2@ncl.ac.uk}\\
	\And
	{Stuart J. Hall} \\
	School of Mathematics, Statistics, and Physics\\
	Newcastle  University\\
	\texttt{stuart.hall@ncl.ac.uk} \\
	\AND
	{Daniel J. Cole}\\
School of Natural and Environmental Sciences\\
Newcastle University\\
	\texttt{daniel.cole@ncl.ac.uk}
}

\renewcommand{\shorttitle}{Geometry and molecular surfaces}

\begin{document}
\maketitle
\noindent
\begin{abstract}

{Ligand-based virtual screening aims to reduce the cost and duration of drug discovery campaigns. Shape similarity can be used to screen large databases, with the goal of predicting potential new hits by comparing to molecules with known favourable properties. This paper presents the theory underpinning RGMolSA, a new alignment-free and mesh-free surface-based molecular shape descriptor derived from the mathematical theory of Riemannian geometry. The treatment of a molecule as a series of intersecting spheres allows the description of its surface geometry using the \textit{Riemannian metric}, obtained by considering the spectrum of the Laplacian. This gives a simple vector descriptor constructed of the weighted surface area and eight non-zero eigenvalues, which capture the surface shape. We demonstrate the potential of our method by considering a series of PDE5 inhibitors that are known to have similar shape as an initial test case. RGMolSA displays promise when compared to existing shape descriptors and in its capability to handle different molecular conformers. The code and data used to produce the results are available via GitHub: \url{https://github.com/RPirie96/RGMolSA}.}

\end{abstract}


\section{Introduction}

The chemical space containing drug-like molecules is vast, with an estimated size of $10^{60}$ molecules \cite{Reymond_2010}. Even with technological advances, it is slow and expensive to screen more than a tiny section of this space experimentally. To address this, chemists have increasingly complemented experimental studies with virtual screening to identify new small molecules (ligands) that might bind to a target, such as a protein, with therapeutic benefit~\cite{Leelananda_Lindert_2016}. Methods that make use of three-dimensional molecular shape have gained traction in recent years due to the importance of shape complementarity between the protein and ligand for strong binding \cite{Kumar_Zhang_2018}. Molecules with known activity can be used to screen large databases to identify other molecules, with similar shapes, which are likely to bind to the same protein targets \cite{Johnson_Maggiora_1990}.\\
\\
Shape-based methods are best known for their ability to identify molecules with the same global shape but that are chemically different from the known active; this phenomenon is known as scaffold hopping. Scaffold hopping can be used to improve drug performance as well as to address unwanted properties, or to generate new intellectual property. An important example of scaffold hopping is provided by the follow-up drugs to Sildenafil, a phosphodiesterase 5 (PDE5) inhibitor used to treat erectile dysfunction~\cite{Cleves_Jain_2008}. The chemical structure of Tadalafil is quite different from that of Sildenafil and Vardenafil (which is a classic ``me-too" drug, where only minor modifications have been made to the original), but the molecule occupies a similar volume in the binding pocket, implying that it has a similar shape (Figure \ref{fig:scafhop}). Sildenafil to Tadalafil is an example of scaffold hopping, where the difference in chemical structure gives a marked improvement in therapeutic profile compared to the other two drugs \cite{Rashid_2005}.\\
\\
Giving a meaningful quantitative measure of how similar two three-dimensional shapes are is an interesting problem in its own right: there is no fixed notion of three-dimensional shape.  Most shape-based virtual screening methods involve using the data describing a molecule to create a mathematical proxy in some latent space and then comparing these `shape descriptors' using a natural underlying geometry of the space; this paper and its sequel \cite{CHP} develop novel shape descriptors based on the Riemannian geometry of the surface of the molecule. \\
\\

\begin{figure}[!htb]
\centering\includegraphics[width=4.5in]{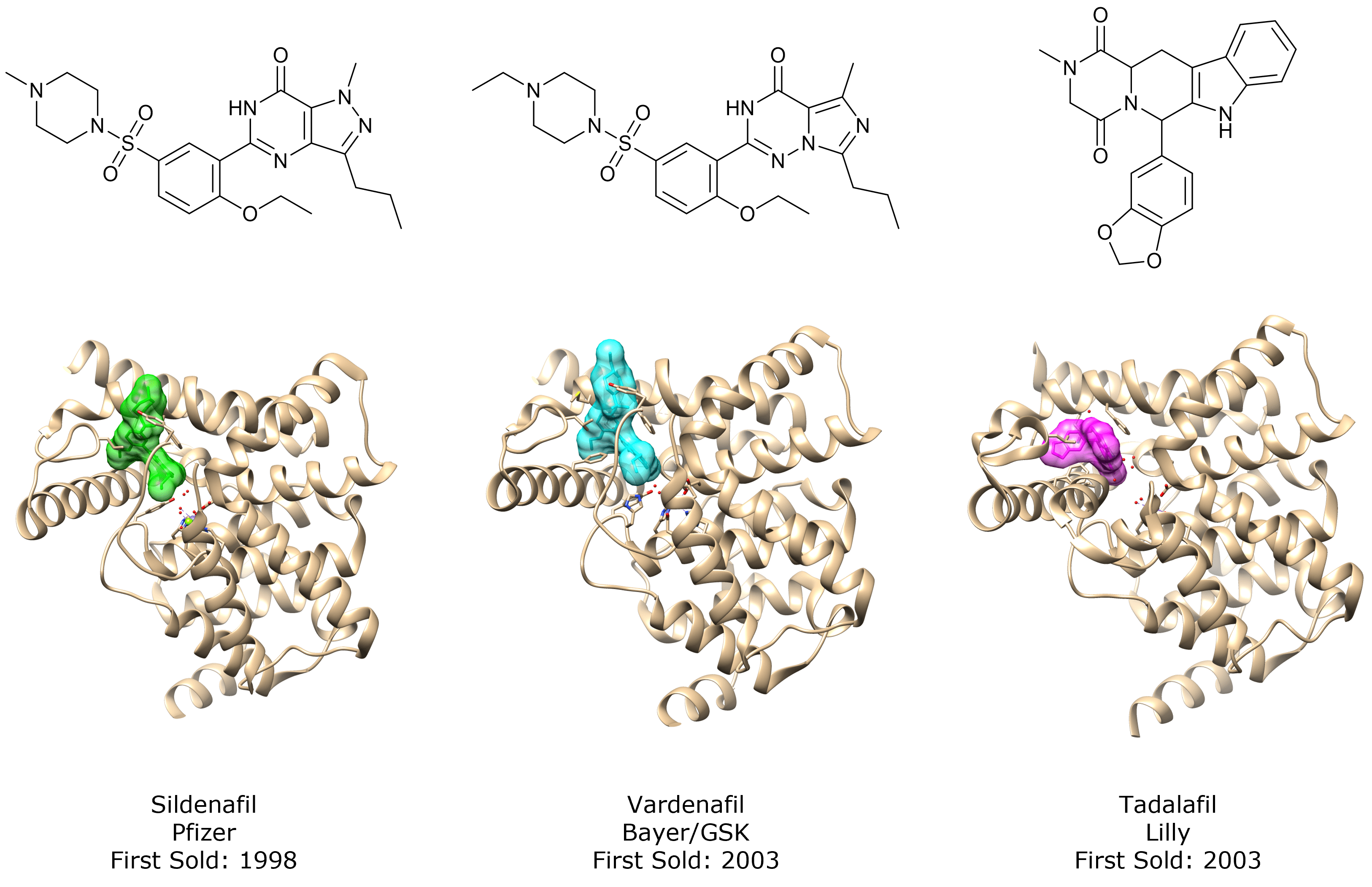}
\caption{PDE5 inhibitors which exemplify the phenomenon of scaffold hopping. Tadalafil (different chemical structure, similar shape) is an example of a scaffold hop from the first in class drug Sildenafil, and offers greatly improved performance, while Vardenafil (a "me-too" follow-up drug) only offers minor improvements.} \label{fig:scafhop}
\end{figure}

Three prominent mathematical methods for comparing shape are as  
follows (Figure \ref{fig:shapesim}): direct comparison of overlap of molecular volumes constructed from Gaussian spheres \cite{Rush_Grant_Mosyak_Nicholls_2005}; vector descriptors constructed based on the distribution of atomic distances within molecules \cite{Ballester_Richards_2007, Schreyer_Blundell_2012,  Shave_Blackburn_Adie_Houston_Auer_Webster_Taylor_Walkinshaw_2015} and vector descriptors based on the consideration of the molecular surface \cite{Max_Getzoff_1988, Novotni_Klein_2003, SCPG}. For further details on existing shape descriptors, we refer the interested reader to reference \cite{Kumar_Zhang_2018}.\\
\\ 
In this paper we focus on a representation derived from the geometry of a molecular surface. While molecules do not have a true surface in the classical way that an apple has a skin, its consideration is still useful in the interpretation of molecular shape and size, as the surface captures many of the same features as the volume, but is less expensive to compute~\cite{Lipkowitz_Boyd_1990}. Although the use of these methods is still in its infancy compared to other shape similarity methods, they offer a compromise between the low computational cost of distance-based descriptors and the accuracy of volume-based methods.

\begin{figure}[!htb] 
\centering\includegraphics[width=5in]{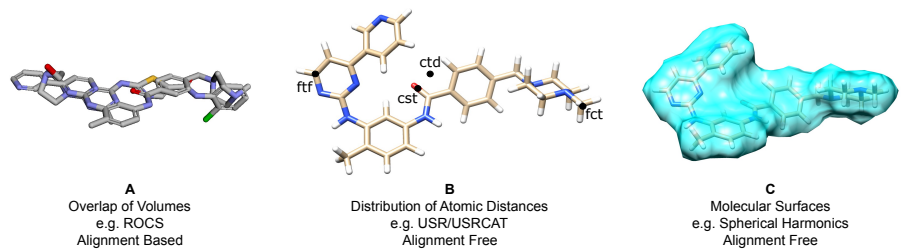}
\caption{The three main approaches to approximating molecular shape: (A) Volume based shape similarity \cite{Rush_Grant_Mosyak_Nicholls_2005}; (B) Atomic-distance based shape similarity \cite{Ballester_Richards_2007}; (C) Surface based shape-similarity \cite{Lipkowitz_Boyd_1990}.} \label{fig:shapesim}
\end{figure}

A well-known method of studying the geometry of a surface (or Riemannian manifold in higher dimensions) is to analyse the behaviour of the \textit{Laplace-Beltrami operator} ($\Delta$) or the \textit{Laplacian} associated with the surface.  More specifically, solving the equation
$$\Delta\varphi = \lambda\varphi,$$
where $\lambda$ is an \textit{eigenvalue} and $\varphi$ is the associated \textit{eigenfunction}, gives an ordered sequence of $\lambda$ known as the \textit{spectrum} of the Laplacian (we use the convention that our eigenvalues are non-negative). For a surface $\mathcal{S} \subset \mathbb{R}^{3}$, the spectrum can be truncated up to the $k^{th}$ eigenvalue producing a vector, 
\[
 v_{\mathcal{S},k}=(\lambda_{1},\lambda_{2},\ldots,\lambda_{k-1},\lambda_{k}),
 \]
that can be used as a shape descriptor. This idea was first introduced by Reuter, Wolter, and Peinecke in Ref.~\cite{RWP} and referred to as `shape-DNA', and was later extended by Seddon \textit{et al.} for molecular shape approximation \cite{SCPG}. We refer the reader to Section \ref{sec:3} for a detailed account of this method.
 
An existing framework that computes the spectrum of the Laplacian (for example, the method used in Seddon \textit{et al.}) involves four steps: 
\begin{enumerate}
    \item The molecular surface is represented by a triangulated mesh (a lattice graph in 3D space composed of N vertices and M edges). 
    \item The Laplace-Beltrami spectrum, $\Delta\varphi = \lambda\varphi$, is then approximated using the finite element method.
    \item Vectors are assigned to each vertex of the surface which describe the surrounding space, giving a local geometry descriptor.
    \item The local descriptors are then clustered to create a global descriptor of shape in order to quantify similarity of two objects. \\
\end{enumerate} 

The method was found to compare well to existing widely used atomic-distance based and volume based methods in a retrospective benchmark study using the Directory of Useful Decoys - Enhanced (DUD-E)~\cite{Mysinger_Carchia_Irwin_Shoichet_2012}. Such studies are used in lieu of an absolute measure of performance, where the usefulness of a shape descriptor is inferred from its ability to place true active molecules higher than decoys in a ranked list. \\
\\
\subsection{Contributions of This Paper}
 Here we propose RGMolSA, an alternative method for approximating the spectrum of the Laplacian that is derived from the mathematical theory of Riemannian Geometry.  By exploiting the fact that the one can represent the molecular surface by a series of intersecting spheres, we give an explicit description of a mathematical object called the \textit{Riemannian metric}. Using this explicit description, we are able to compute in a closed form the integrals that one needs to calculate in the standard method for approximating the spectrum of the Laplacian.  This removes the need to compute the mesh, which we hypothesise will increase the speed and give a more finely tuned description of the surface.  An overview of the steps required for computing the vector of nine eigenvalues used to represent the surface is given in Figure \ref{fig:silflow}. Due to their known similarity, the PDE5 inhibitors in Figure \ref{fig:scafhop} will be used throughout as proof of concept, and to discuss the dependence of the results on factors such as the different conformers a molecule can adopt.\\
\\
One possible deficiency in our approximation of the spectrum is that it involves the choice of a base atom; the descriptor provides a good representation of the geometry of the surface near to this base atom but the atoms further away are `higher frequency' objects and we would need many eigenvalues to be able to describe them accurately (as mentioned, our descriptor approximates the spectrum up to the ninth eigenvalue but in theory the method could be used to compute more).  In the sequel to this paper, we use the Riemannian metric we have computed in the theory of K\"ahler quantisation to produce a completely novel shape descriptor which lies in the manifold $GL(N,\mathbb{C})/U(N)$;  this manifold-valued descriptor is global.  The reader should regard this paper as laying the foundations for the use of these methods from Riemannian and complex geometry in molecular similarity searching; in Section \ref{section:conclusion} we discuss possible ways of refining and developing our geometric shape descriptors.     

\begin{figure}[!htb] 
\centering\includegraphics[width=4.5in]{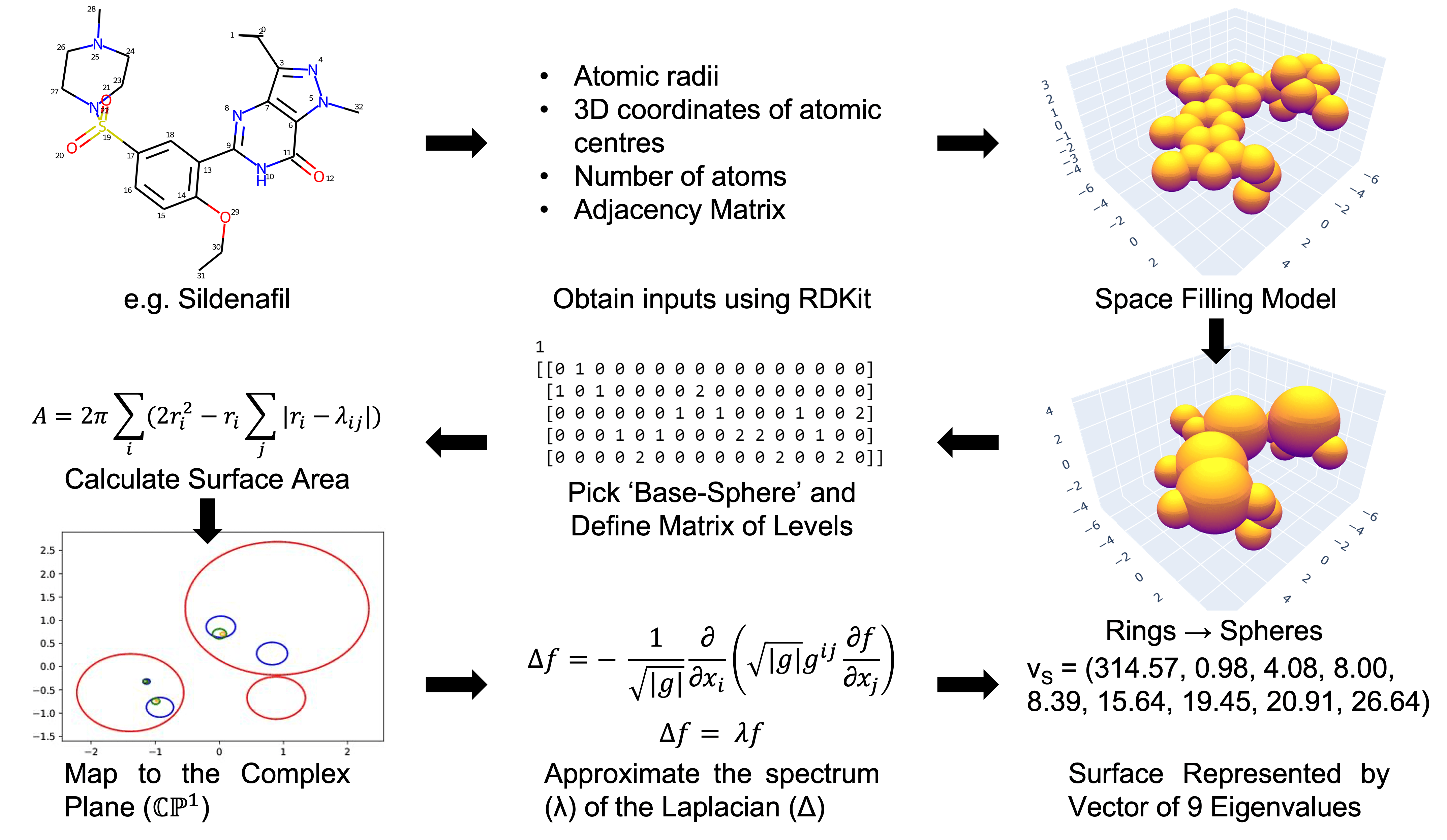}
\caption{Key steps involved in the computation of our shape descriptor using Sildenafil as an example.} \label{fig:silflow}
\end{figure}

\section{The Mathematical Description of Shape: Riemannian Geometry}  \label{section:riemanniangeometry}
\subsection{The Riemannian Metric} \label{section:riemannmetric}
We will regard the molecular surface $\mathcal{S}$ as a mathematical surface (or a closed orientable two-dimensional manifold in more technical language). We refer the reader to Chapter 3 of the book \cite{DonRS} for a technical definition. Surfaces are classified by their genus, which is a topological invariant that counts the number of `holes' the surface possesses. {\bf The first, and most important, assumption we make in this paper is that the surface has genus zero which means that there is a continuous one-to-one map from the surface $\mathcal{S}$ onto the unit sphere $\mathbb{S}^{2}$}. This might seem a problematic assumption from the point of view of considering molecules containing rings such as benzene or pyridine, but we shall address this difficulty in Section \ref{section:moldata}. One limitation introduced by this assumption is the inability to describe macrocyclic molecules (large rings containing more than twelve atoms), which have a genuine hole in their surface \cite{Yudin_2015}. While it would be useful for chemists to be able to deal with macrocycles, such molecules do not feature heavily in small-molecule drug design.\\
\\
Part of the mathematical definition of a surface $\mathcal{S}$ involves the idea of a `coordinate chart' whereby the surface is described by continuous maps 
$\psi:U\rightarrow V\subset \mathcal{S}$, where $U \subset \mathbb{R}^{2}$ is an open subset of the plane $\mathbb{R}^{2}$ and $V$ is an open subset of the surface $\mathcal{S}$ (one can think of $V$ as the intersection of the surface $\mathcal{S}$ with an open set in three-dimensional Euclidean space $\mathbb{R}^{3}$). Each map furnishes the subset $V\subset \mathcal{S}$ with coordinates $(x,y) \in U$.  In practice, we will see how to parameterise the entire molecular surface (except for one point) by a single map $\psi:\mathbb{R}^{2} \rightarrow \mathcal{S}$; we will identify $\mathbb{R}^{2}$ with the complex plane $\mathbb{C}$.\\
\\
The mathematical theory of shape we will use is that of Riemannian geometry where the shape of a surface is encoded by a \textit{Riemannian metric}. In the case of a surface $\mathcal{S}$, one can think of the metric, locally at least, as positive-definite symmetric $2\times 2$ matrix 
\begin{equation} \label{eqn:met_mat}
g= \left( \begin{array}{cc} a_{11} & a_{12}\\
a_{12} &a_{22}
\end{array}
\right),
\end{equation}
where each of the $a_{ij}$ are locally-defined smooth functions on $\mathcal{S}$ (here locally-defined means we can use local coordinate charts to view them as functions in the coordinates). Rather than writing metrics in the form of a matrix of functions as in Equation (\ref{eqn:met_mat}), we will often find it convenient to record the metric in the `differential-form notation' and write
\[
g = a_{11}dx^{2}+a_{12}dx\cdot dy+a_{22}dy^{2}.
\]
Given an immersion (a differentiable function such that the derivative has full rank) ${F:U\subset\mathbb{R}^{2}\rightarrow V\subset\mathcal{S} \subset \mathbb{R}^{3}}$, one obtains an induced metric on $U$ (often identified with its image in $\mathcal{S}$) by `pulling-back' the Euclidean dot-product on $\mathbb{R}^{3}$; the pull-back of the Euclidean metric by $F$ is denoted $F^{\ast}(g_{Euc})$. 
We can write $F$ as
$$ F(x,y) = (F_{1}(x,y),F_{2}(x,y),F_{3}(x,y)),$$
where $F_{i}:U\rightarrow\mathbb{R}$ are differentiable functions. The induced metric then takes the form
\begin{equation}\label{I_Form}
g=F^{\ast}(g_{Euc})=\left( \begin{array}{cc} \frac{\partial F}{\partial x}\cdot \frac{\partial F}{\partial x}& \frac{\partial F}{\partial x}\cdot\frac{\partial F}{\partial y}\\
\frac{\partial F}{\partial x}\cdot\frac{\partial F}{\partial y}& \frac{\partial F}{\partial y}\cdot\frac{\partial F}{\partial y}
\end{array}
\right),
\end{equation}
where 
$$\frac{\partial F}{\partial x}:=\left( \frac{\partial F_{1}}{\partial x},\frac{\partial F_{2}}{\partial x},\frac{\partial F_{3}}{\partial x}\right),$$
and 
$$\frac{\partial F}{\partial y}:=\left( \frac{\partial F_{1}}{\partial y},\frac{\partial F_{2}}{\partial y},\frac{\partial F_{3}}{\partial y}\right),$$
and $\cdot$ denotes the Euclidean dot-product in $\mathbb{R}^{3}$.\\
\\
In this case, the metric $g$ is often called the $1^{st}$-fundamental form of $\mathcal{S}$. From the metric, one can define and compute geometric quantities such as distances between points or the curvature at a point. Thus our first task is to give a concrete description of the metric associated to a molecular surface. We do this by constructing a particular immersion and then computing the $1^{st}$-fundamental form.

\subsection{The Data Defining the Molecular Surface} \label{section:moldata}
Ignoring ring structures for now, we consider a molecular surface as being formed from $N$ intersecting spheres. As the positions of hydrogen atoms are well-defined and they are comparatively small in size, they are typically excluded to simplify the problem without much loss in descriptor quality. The quantity $N$ is therefore the number of heavy (non-hydrogen) atoms in the molecule. The data describing the surface are $N$ pairs $\{(c_{i},r_{i})\}_{i=1}^{N}$ where $c_{i}$ is the centre of the $i^{th}$ sphere in $\mathbb{R}^{3}$ and $r_{i}$ is the  radius. We also need an $N\times N$ adjacency matrix $T$
where 
\[
T_{ij}=\left\{\begin{array}{cc} 
1 & \textrm{if spheres {\it i} and {\it j} intersect}\\
0 & \textrm{otherwise (or }i=j) .
\end{array}\right.
\]
The constituent atoms of each ring within the molecule are then replaced with a single sphere, allowing it to be treated as having a genus of zero. To do this in practice, the centres, radii and intersections of each atom making up the ring must be accounted for. The centre is taken as the centre of mass of the contributing atoms and the collection of radii are substituted with a fixed value of 2.25 {\AA}. This value is derived from the average surface area of common rings found within drug molecules, which are typically 4-7 atoms in size and adjusted slightly to ensure the correct intersections between spheres still occur. This choice may require further tuning to improve the accuracy of the model. As the shape descriptors we employ in this paper and the sequel represent an approximation to the molecular surface, and as the rings are all relatively similar in size,  we do not believe it to be necessary to differentiate the radii of different rings. The adjacency matrix $T$ is also updated to include the intersection of the single sphere with its neighbours.\\ 
\\
This data allows for a straightforward calculation of the surface area. The following is presumably well-known, but we will give a short proof for completeness. 

\begin{proposition}
Let $\{(c_{i},r_{i})\}_{i=1}^{N}$ and $T$ be the data describing the molecule as outlined previously. Then the surface area is given by the formula
\begin{equation}\label{eqn:surf_area}
\mathrm{Area} = 2\pi \sum_{i}\left(2r_{i}^{2}-\left(r_{i}\sum_{j}T_{ij}|r_{i}-\lambda_{ij}|\right)\right),
\end{equation}
where 
\begin{equation}\label{eq:lambda}
\lambda_{ij} =\left\{\begin{array}{cc} \frac{r_{i}^{2}-r_{j}^{2}+\|c_{i}-c_{j}\|_{\mathrm{Euc}}^{2}}{2\|c_{i}-c_{j}\|_{\mathrm{Euc}}} & \mathrm{if} \qquad i\neq j, \\
0 & \mathrm{if} \qquad i = j,
\end{array}\right. 
\end{equation}
and 
\[
 \|x-y\|_{\mathrm{Euc}} :=\sqrt{\sum_{i=1}^{3}(x_{i}-y_{i})^2},
 \]
 for $x,y, \in \mathbb{R}^{3}$.
\end{proposition}

\begin{proof}
To find the molecular area we find the contribution of each sphere by $4\pi r_{i}^{2} - K_{i}$, where $K_{i}$ is the area missing due to the $i^{th}$ sphere intersecting its neighbours, and then sum over the $i$. If spheres $i$ and $j$ intersect then it will be useful to find the equation of the plane of intersection of the spheres $i$ and $j$. If we consider the line joining the centres of the $i^{th}$ and $j^{th}$ spheres, then the plane passes through a point $p_{ij}$ on this line given by
\[
p_{ij} = \lambda_{ij}\left(\frac{c_{j}-c_{i}}{\|c_{i}-c_{j}\|_{\mathrm{Euc}}}\right)+c_{i},
\]
for some value of $\lambda_{ij} >0$. Elementary trigonometric considerations yield the formula
\[
\lambda_{ij} = \frac{r_{i}^{2}-r_{j}^{2}+\|c_{i}-c_{j}\|_{\mathrm{Euc}}^{2}}{2\|c_{i}-c_{j}\|_{\mathrm{Euc}}}.
\]
The equation of the plane is thus $x\cdot (c_{i}-c_{j}) = p_{ij}\cdot(c_{i}-c_{j})$.
The area of the $i^{th}$ sphere that the $j^{th}$ sphere `removes' is thus
$$ 
2\pi r_{i}|r_{i}-\lambda_{ij}|.
$$
The result follows by summing over $j$ to obtain the term  
\[
K_{i} = 2\pi r_{i}\sum_{j}T_{ij}|r_{i}-\lambda_{ij}|.
\]
(Note in the formula (\ref{eqn:surf_area}) we weight the $T_{ij}$ entry of the adjacency matrix so that we only remove terms involving genuine neighbours). The formula follows by summation over $i$.
\end{proof}

We record the area as the first geometric invariant of the surface and then re-scale the data so that the resulting surface has area $4\pi$ (the area of the standard unit sphere $\mathbb{S}^{2}$). One facet of Riemannian geometry is that two metrics which differ only in scale are viewed as having essentially the same shape. Re-scaling so that all molecules are treated as having the same surface area addresses this. {The fixed size of the radii of atoms mean this phenomenon would not be physically possible for molecules; however, as re-scaling significantly simplifies the subsequent mathematics it is still included in the generation of our descriptor. A molecule of twice the size as the one we are comparing to will be unable to fit in the protein binding pocket of interest, and is therefore not a useful suggestion as a potential similar molecule. We account for the re-scaling in our descriptor by replacing the first eigenvalue (which is always zero) with the original, unscaled surface area, weighted so as not to dominate in the similarity calculation (see Section \ref{section:SAweight}).} 

\subsection{Piecewise Stereographic Projection}

To describe the surface, we will construct a map (which we call \textit{piecewise stereographic projection})
\[
\Phi_{ps}:\mathbb{C}\rightarrow \mathcal{S} \subset\mathbb{R}^{3}.
\]
By adding a `point at infinity', it is standard to identify $\mathbb{S}^{2}$ with ${\mathbb{C}\cup\{\infty\}}$ (the latter set is often called the Riemann sphere \cite{DonRS}). This identification is not canonical but any two identifications differ by an action of the automorphism group of the Riemann sphere, $PSL(2,\mathbb{C})$. We shall see later in the article how various choices in our method are covered by this action. We can then extend the map $\Phi_{ps}$ to establish a diffeomorphism between $\mathcal{S}$ and the sphere $\mathbb{S}^{2}$.\\
\\
The construction of the map $\Phi_{ps}$ requires the choice of a starting atom which we will refer to as the base sphere, or `level-$0$' sphere. We will choose the atom closest to the centre of mass by first finding the centroid of the molecule and then taking the atom with the smallest Euclidean distance from this point.
This inductively leads to the notion of a level-$k$ sphere for $k>0$:
\[
\textrm{a level-$k$ sphere is one which intersects level-$(k-1)$ sphere.}
\]
We further divide the level-$k$ spheres into two types:
\begin{itemize}
\item a \textit{terminal} level-$k$ sphere only intersects a level-$(k-1)$ sphere,
\item a \textit{non-terminal} level-$k$ sphere intersects a level-$(k+1)$ sphere.
\end{itemize}
The map $\varphi_{r}:\mathbb{C}\rightarrow \mathbb{R}^{3}$ given by
\[
\varphi_{r}(z) = \frac{r}{1+|z|^{2}}(2\mathrm{Re}(z),2\mathrm{Im}(z),|z|^{2}-1),
\]
is the standard stereographic projection from the complex plane onto a sphere of radius $r$  centred at the origin; the image of the map does not contain the `north pole' $(0,0,r)$.  It is useful to note that a truncated sphere, with $x_{3}\leq h$ for $h\in (-r,r)$, is the image under $\varphi_{r}$ of the disc 
\[
\mathbb{D}(\mathcal{R})  = \left\{z\in \mathbb{C} \ : \ |z|\leq \mathcal{R}\right\}\subset \mathbb{C},
\]
of radius
\[
\mathcal{R} = \sqrt{\left(\frac{r+h}{r-h} \right)}.
\]
By re-scaling $z\rightarrow \tau z$, where $\tau>0$, we can map, via $\varphi_{r}$, a disc of an arbitrary radius onto a truncated sphere of radius $r$ and of arbitrary height.\\
\\
As a model case, we consider two intersecting spheres of radii $r_1$ and $r_2$ with the centre of the second sphere at the origin and the centre of the first sphere at $\mathbf{c} =(0,0,c)$ with $c>0$. The height of each sphere is given by
\[
h_{i} = 2r_{i}-|r_{i}-\lambda_{ij}|.
\]
where $\lambda_{ij}$ is given by Equation (\ref{eq:lambda}). The first sphere is the image under the map $\varphi_{r_{1}}+\mathbf{c}$ of the complement of the disc $\mathbb{D}(R_{1})$ where
\[
R_{1} = \sqrt{\frac{2r_{1}-h_{1}}{h_{1}}}.
\]
Put more succinctly, the first sphere is
\[
\{\varphi_{r_{1}}(z)+c : z\in \mathbb{C} \backslash \mathbb{D}(R_{1})\} \subset \mathbb{R}^{3}.
\]
The second sphere is the image under $\varphi_{r_{2}}$ of the disc $\mathbb{D}(R_{2})$ where
\[
R_{2} = \sqrt{\frac{h_{2}}{2r_{2}-h_{2}}},
\]
that is, the second sphere is
\[
\{\varphi_{r_{2}}(z) : z\in  \mathbb{D}(R_{2})\} \subset \mathbb{R}^{3}.
\]
Hence we define the \textit{2-atom piecewise stereographic projection} map
\[
\Phi_{2A}(z) = \left\{\begin{array}{cc}
\frac{r_{1}}{1+|\tau z|^{2}} ({2\mathrm{Re}(\tau z)},2\mathrm{Im}(\tau z),|\tau z|^{2}-1+c) & \mathrm{if} \ |z|\geq R_{2},\\
\frac{r_{2}}{1+|z|^{2}} ({2\mathrm{Re}(z)},2\mathrm{Im}(z),|z|^{2}-1) & \mathrm{if} \ |z|<R_{2},
 \end{array}\right.
\]
where $\tau=\dfrac{R_{1}}{R_{2}}$.\\
\\
We can compute explicitly the metric $\Phi_{2A}^{\ast}(g_{Euc})$; if we identify $z=x+\sqrt{-1}y$, then
\[
\Phi_{2A}^{\ast}(g_{Euc}) =  \left\{\begin{array}{cc}
\dfrac{4r_{2}^{2}(dx^{2}+dy^{2})}{(1+|z|^{2})^{2}} & \mathrm{if} \ |z|<R_{2},\\
\\
\dfrac{4\tau^{2}r_{1}^{2}(dx^{2}+dy^{2}) }{(1+|\tau z|^{2})^{2}} &  \mathrm{if} \ |z|\geq R_{2}.\end{array} \right.
\]
The surface is built from the intersection of truncated spheres after they have been rotated and translated into the correct position; we now consider how rotation interacts with stereographic projection. \\
\\
Rotations about an axis through the origin in three-dimensional Euclidean space $\mathbb{R}^{3}$ can be encoded by a $3\times 3$  matrix $M$ (the rotation sends a point $p\in \mathbb{R}^{3}$ to $Mp$).  The following matrix rotates the vector $(0,0,1)$ onto the vector $v=(v_{1},v_{2},v_{3})$ (where we assume $v \neq (0,0,-1)$)
\[
M_{rot} = \left( \begin{array}{ccc}
1-\frac{v_{1}^{2}}{1+v_{3}} & - \frac{v_{1}v_{2}}{1+v_{3}} & v_{1}\\
- \frac{v_{1}v_{2}}{1+v_{3}} & 1-\frac{v_{2}^{2}}{1+v_{3}} & v_{2}\\
-v_{1} & -v_{2} & v_{3}
\end{array}
\right).
\]

The inverse of the map $\varphi_{r}$ is the map
${\varphi_{r}^{-1}:\mathbb{S}_{r}^{2}\backslash\{(0,0,r)\}\rightarrow \mathbb{C}}$ given by
\[
\varphi_{r}^{-1}(x_{1},x_{2},x_{3}) = \frac{x_{1}}{r-x_{3}}+\sqrt{-1}\frac{x_{2}}{r-x_{3}}.
\]
(If we wish to extend to the whole sphere, this map sends `the North Pole' $N=(0,0,r)$ to the point $\infty$). 

The rotation $M_{rot}$ induces a map $\eta: {\mathbb{C}}\rightarrow {\mathbb{C}}$ that makes the following diagram commute:
\begin{large}
\[
\begin{tikzcd}
  \mathbb{S}^{2}_{r} \arrow[r, "M_{rot}"] \arrow[d, "\varphi_{r}^{-1}"]
    & \mathbb{S}^{2}_{r}  \arrow[d, "\varphi_{r}^{-1}"] \\
  {\mathbb{C}} \arrow[r, black, "\eta" black]
& {\mathbb{C}} \end{tikzcd}
\]
\end{large}
The following can be proved by elementary algebra and we omit the proof.
\begin{lemma}\label{lem:rot_lem}
 The induced map $\eta:\mathbb{C}\rightarrow\mathbb{C}$ is given by the M\"obius transformation
\[
\eta(w) = \frac{\alpha w+\beta}{\bar{\alpha} -\bar{\beta}w},
\]
with $|\alpha|^{2}+|\beta|^{2}=1$.  The constants $\alpha \in \mathbb{R}$ and $\beta\in \mathbb{C}$ are given by
\begin{eqnarray}
\alpha = \sqrt{\left( \frac{1+v_{3}}{2} \right)},\\
\beta = -\sqrt{\left(\frac{1}{2(1+v_{3})} \right)}(v_{1}+\sqrt{-1}v_{2}).
\end{eqnarray}
Furthermore, the image under $\eta$ of a disc $\mathbb{D}(R)$ of radius $R>0$ and centre $0$ is give by the set of $w \in \mathbb{C}$ such that
\[
\bigg|w-\frac{(1+R^{2})\alpha\beta}{\alpha^{2}-R^{2}|\beta|^{2}}\bigg| \sim \frac{R^{2}}{(\alpha^{2}-R^{2}|\beta|^{2})^{2}},
\]
where $\sim$ is $\leq$ if $(\alpha^{2}-R^{2}|\beta|^{2})>0$ and $\sim$ is $\geq$ if $(\alpha^{2}-R^{2}|\beta|^{2})<0$.
\end{lemma}

To obtain the fundamental form of a general pair of intersecting spheres we use the map $\eta$.
Viewing $\eta^{-1}$ as a map $\eta^{-1}:{\mathbb{C}}\rightarrow ({\mathbb{C}},\Phi_{2A}^{\ast}(g_{Euc}))$, that is as a map from an abstract copy of $\mathbb{C}$ onto the copy endowed with the 2-atom metric $\Phi^{\ast}_{2A}(g_{Euc})$, we can again pullback by $\eta^{-1}$ to get the induced `skew' 2-atom metric. Using Lemma \ref{lem:rot_lem} see that the metric $g=\eta^{-1 \ast}(\Phi_{2A}^{\ast}(g_{Euc}))$ can be written as
\begin{equation*}
g = \left\{\begin{array}{cc}
\dfrac{4r_{2}^{2}(dx^{2}+dy^{2})}{(1+|z|^{2})^{2}} & \mathrm{if} \  \bigg|z-\dfrac{(1+R_{2}^{2})\alpha\beta}{(\alpha^{2}-R_{2}^{2}|\beta|^{2})}\bigg|^{2} \geq  \dfrac{R_{2}^{2}}{(\alpha^{2}-R_{2}^{2}|\beta|^{2})^{2}},\\
\\
\dfrac{4\tau^{2}r_{1}^{2}(dx^{2}+dy^{2})}{(|\bar{\beta}z+\alpha|^{2}+\tau^{2}|\alpha z-\beta|^{2})^{2}} & \mathrm{if} \ \bigg|z-\dfrac{(1+R_{2}^{2})\alpha\beta}{(\alpha^{2}-R_{2}^{2}|\beta|^{2})}\bigg|^{2} \leq  \dfrac{R_{2}^{2}}{(\alpha^{2}-R_{2}^{2}|\beta|^{2})^{2}},\\
\end{array} \right.
\end{equation*}
where $R_{1},R_{2},\tau$ are all given as in the preceding discussion.\\
\\
It will be useful to note
\[
\dfrac{2\tau^{2}r_{1}^{2}}{(|\bar{\beta}z+\alpha|^{2}+\tau^{2}|\alpha z-\beta|^{2})^{2}} = \dfrac{2(|\beta|^{2}+\tau^{2}\alpha^{2})^{-2}\tau^{2}r_{1}^{2}}{(|z-D_{1}|^{2}+D_{2})^{2}},
\]
where
\begin{equation*}
D_{1} = \dfrac{(\tau^{2}-1)\alpha\beta}{(|\beta|^{2}+\tau^{2}\alpha^{2})} \qquad \mathrm{and} \qquad D_{2} = \frac{\tau^{2}}{(|\beta|^{2}+\tau^{2}\alpha^{2})^{2}}
\end{equation*}
Actually, the precise expression here is not what is important; what we really wish to highlight is that the metric is of the form
\begin{equation}\label{eqn:met_form}
g = \left\{\begin{array}{cc}
\dfrac{4r_{2}^{2}(dx^{2}+dy^{2})}{(1+|z|^{2})^{2}} & \mathrm{if} \  z\not \in \mathbb{D}(a,R),\\
\\
\frac{C(dx^{2}+dy^{2})}{(|z-A|^{2}+B)^{2}} & \mathrm{if} \  z \in \mathbb{D}(a,R),
\end{array} \right.    
\end{equation}

where $R,B,C>0$ and $a,A\in \mathbb{C}$ (here $\mathbb{D}(a,R)$ is the disc of radius $R$ centred at $a\in \mathbb{C}$). The values of $a,R,A,B$ and $C$ are all computable from the original data in an explicit, albeit complicated, manner.\\ 
\\
In our language, the second sphere is our level-$0$ sphere centred at the origin and the first sphere is a (terminal) level-1 sphere.  As the induced maps $\eta$ act as isometries on the level-0 sphere (i.e. pullback preserves the form of the metric) we can immediately extend this discussion to other level-1 spheres.  Each one will have an associated disc $\mathbb{D}(a,R)$ where the metric takes the form
\[
\frac{C}{(|z-A|^{2}+B)^{2}}(dx^{2}+dy^{2}).
\]
To extend to level-2 spheres and beyond we repeat the process inductively. Geometrically we can rotate a given level-1 sphere so that its centre is on the positive $x_{3}$ axis and then translate so that it is centred at the origin.  We can then treat the level-2 spheres as level-1 spheres relative to it and proceed as previously. In the complex plane, this amounts to mapping the disc $\mathbb{D}(a,R)$ corresponding to the sphere  to the origin by the map $\eta^{-1}$ (where $\eta$ is induced by the rotation) and then re-scaling the disc by $\tau$.  This produces a new coordinate $\xi=\tau\eta^{-1}(z)$ where the metric has the form (\ref{eqn:met_form}) in $\xi$. In the plane $\mathbb{C}$, the level-2 sphere will correspond to a disc $\mathbb{D}(a',R') \subset \mathbb{D}(a,R)$ where the metric has the form  
\[
\frac{C'}{(|z-A'|^{2}+B')^{2}}(dx^{2}+dy^{2}).
\]
The process of taking a disc to the origin by a M\"obius transformation, performing stereographic projection, then rotating and translating the truncated sphere into position is how we obtain the map $\Phi_{ps}$.\\
\\
To implement the piecewise stereographic projection in Python, it is convenient to ensure the `north pole' (the point on the surface not mapped to by $\Phi_{ps}$) is in the level-0 sphere and so is not covered by a higher-level sphere.  This can be done by rotating the initial data if, after centering the level-0 at the origin, the point $(0,0,r_{B})$ is not part of the surface.\\
\\
Putting all this together, we record the results of this section as a theorem.
\begin{theorem}
The metric $g=\Phi_{ps}^{\ast}(g_{Euc})$ induced by the mapping $\Phi_{ps}:\mathbb{C}\rightarrow \mathcal{S}\subset \mathbb{R}^{3}$ is given by
\begin{equation}\label{eqn:metric_form_gen}
g = \left\{\begin{array}{cc}
\frac{4r_{B}^{2}}{(1+|z|^{2})^{2}}(dx^{2}+dy^{2}) & \mathrm{if} \  z \in\mathcal{C}\\ 
& \\
\frac{C_{1}}{(|z-A_{1}|^{2}+B_{1})^{2}}(dx^{2}+dy^{2}) & \mathrm{if} \  z \in \mathbb{D}(a_{1},R_{1}),\\
 & \\
\frac{C_{2}}{(|z-A_{2}|^{2}+B_{2})^{2}}(dx^{2}+dy^{2}) & \mathrm{if} \  z \in \mathbb{D}(a_{2},R_{2}),\\
\vdots & \vdots \\
\frac{C_{N-1}}{(|z-A_{N-1}|^{2}+B_{N-1})^{2}}(dx^{2}+dy^{2}) & \mathrm{if} \  z \in \mathbb{D}(a_{N-1},R_{N-1}),\\
\end{array}\right.
\end{equation}
where $r_{B}$ is the radius of the base sphere and 
\[
\mathcal{C} = \mathbb{C}\backslash \mathbb{D}(a_{1},R_{1})\cup \mathbb{D}(a_{2},R_{2})\cup \cdots \cup \mathbb{D}(a_{N-1},R_{N-1}),
\]
is the complement of the discs $\mathbb{D}(a_{1},R_{1}), \ldots, \mathbb{D}(a_{N-1},R_{N-1})$ which corresponds to the points in the base sphere.
\end{theorem}
We have now used the initial data defining the surface (the set of centres and radii $\{(c_{i},r_{i})\}$ and the adjacency matrix $T$) to `unwrap` the surface onto the complex plane $\mathbb{C}$ where its geometry is encoded by the $a_{i},A_{i} \in \mathbb{C}$ and $C_{i},B_{i},R_{i}>0$, each of these quantities being computable from the original data.\\
\\
The form of the metric in the previous theorem will not come as a surprise to a reader who is familiar with Riemannian geometry; the metric locally has the form of the round metric once suitably rescaled and translated. This is exactly the geometric description of a surface formed by intersecting spheres. What the above discussion should serve as is a recipe for computing the quantities $a_{i}$, $A_{i}$, etc. using the original molecular data and the induced maps in $\mathbb{R}^{2}$. We also remark that the metric is not smooth but rather continuous.  This is because of the `corners' formed where the spheres intersect.  Though much of theory discussed in the next section is developed for smooth metrics, this is not a major issue and the Rayleigh Ritz approximation does not require the calculation of any derivatives of the metric.

\section{Approximating the Spectrum of the Laplace--Beltrami Operator: A Mesh-free Approach} \label{sec:3}
\subsection{Background on the Laplace--Beltrami operator and its spectrum}
In this section we frame our discussion in terms of a general manifold $M$; for readers unfamiliar with this notion, one can replace $M$ with the molecular surface $\mathcal{S}$ from the previous section.\\ 
\\
Associated to any smooth Riemannian manifold $(M,g)$ is an operator $\Delta:C^{\infty}(M)\rightarrow C^{\infty}(M)$ known as the \textit{Laplace--Beltrami operator} or the \textit{Laplacian}. There are a number of definitions of the Laplacian; one uses local co-ordinates and is given by
\[
\Delta f :=-\frac{1}{\sqrt{|g|}}\frac{\partial}{\partial x_{i}}\left(\sqrt{|g|}g^{ij}\frac{\partial f}{\partial x_{j}} \right),
\]
where $|g|$ is the determinant of the metric $g$ in the $x$-coordinate system and $g^{ij}$ is the matrix inverse of $g$ in these coordinates.  This formula can be thought of as a generalisation of the `flat' Laplacian of Euclidean space
\[
\Delta_{Euc} := -\sum_{i}\frac{\partial^{2}}{\partial x_{i}^{2}}
\]
to the curved geometry represented by the manifold $(M,g)$. It is natural to look for solutions to the eigenvalue equation 
\begin{equation}\label{eqn:evalue}
    \Delta f =\lambda f,
\end{equation}
for some $\lambda \in \mathbb{R}$ and $f \in C^{\infty}(M)$. A foundational result in Riemannian geometry (see Theorem 1.29 in \cite{RosLM} for example) is that, in the case when the manifold $M$ is compact, the set of eigenvalues of $\Delta$ is a discrete set that accumulates only at infinity.  Put another way, we can order the $\lambda$ as a sequence
\[
0=\lambda_{0}<\lambda_{1}\leq \lambda_{2}\leq \lambda_{3}\leq \ldots
\]
 that for any $\Lambda \in \mathbb{R}$, there are only finitely many $\lambda_{i}\leq \Lambda$.  The set $\{\lambda_{i}\}$ is known as the \textit{spectrum} of the Laplacian $\Delta$.  The field of spectral geometry concerns the question of what geometric information is contained in the set $\{\lambda_{i}\}$.  For example, it is known that two geometrically distinct Riemannian manifolds $(M_{1},g_{1})$ and $(M_{2},g_{2})$ can yield the same spectrum (such metrics are referred to as isospectral)  but such manifolds must share some of the same coarse geometric properties such as dimension and volume. The spectrum can be truncated up to the $k^{th}$ eigenvalue to give a vector
 \[
 v_{\mathcal{S},k}=(\lambda_{1},\lambda_{2},\ldots,\lambda_{k-1},\lambda_{k})
 \]
 describing the shape of the surface \cite{RWP}. \\
 \\
The truncated spectrum $v_{\mathcal{S},k}$ is, in theory, invariant under isometric deformation; it appears therefore that no pre-alignment step is needed to produce an optimum similarity calculation (as is needed for several existing shape descriptors). However, it is usually impossible to determine the spectrum exactly and so an approximation is calculated.  The Rayleigh--Ritz method we describe in the next section will produce a vector that is invariant under rotations and translations of the surface but does depend upon the choice of level-0 sphere (see the discussion in Section \ref{section:conclusion}).  However, for large values of $k$, we expect the smaller eigenvalues $\lambda_{0},\lambda_{1},\ldots,$ etc. to be well approximated and so not as dependent on this choice.\\
\\
At the heart of most approximation schemes is the generation of a mesh - either to compute the Laplacian directly and form a large system of linear equations (e.g. the approach used in \cite{SCPG}) or to provide the points in a cubature scheme for calculating integrals used in the Rayleigh--Ritz method.  We will demonstrate a method that does not need this step by computing such integrals explicitly.\\
\\
\vspace{-30pt}
\subsection{Approximating the Spectrum}
We describe an approach that is often called the \textit{Rayleigh--Ritz} approximation. We select $n$ arbitrary trial functions $f_{1},f_{2},\ldots,f_{n} \in C^{\infty}(M)$ and consider the $n$-dimensional vector space $\mathcal{V}\subset C^{\infty}(M)$ given by their linear span. Let $f$ be an eigenfunction with eigenvalue $\lambda$ and suppose that $f\in \mathcal{V}$. Then
for any test function $f_{j}$ we must have
\begin{equation}\label{eqn:Ray_Rit}
\lambda \langle f,f_{j}\rangle_{L^{2}(M)} = \lambda\int_{M}f\cdot f_{j} \ dV_{g} = \int_{M} (\Delta f) \cdot f_{j} \ dV_{g}=\int_{M} g(\nabla f,\nabla f_{j}) \ dV_{g} = \langle\nabla f, \nabla f_{j}\rangle_{L^{2}}(M),
\end{equation}
where the penultimate inequality follows by the manifold analogue of integration-by-parts. Using the assumption that $f\in \mathcal{V}$ means we can write 
\[
f=\sum_{i=1}^{n}\varepsilon_{i}f_{i},
\]
for coefficients $\varepsilon_{i}\in \mathbb{R}$. Substituting this into Equation (\ref{eqn:Ray_Rit}) yields
\[
\lambda\sum_{i=1}^{n}\varepsilon_{i}\langle f_{i},f_{j}\rangle_{L^{2}(M)} = \sum_{i=1}^{n}\varepsilon_{i}\langle \nabla f_{i},\nabla f_{j} \rangle_{L^{2}(M)}.
\]
This can be written more succinctly by defining the symmetric matrices $\mathcal{A},\mathcal{B}$ by 
\[
\mathcal{A}_{ij} = \langle f_{i},f_{j}\rangle_{L^{2}(M)} \qquad \mathrm{and} \qquad \mathcal{B}_{ij} =  \langle \nabla f_{i},\nabla f_{j}\rangle_{L^{2}(M)}.
\]
We then see that Equation (\ref{eqn:Ray_Rit}) is really equivalent to 
\[
\mathcal{A}^{-1}\mathcal{B}\varepsilon = \lambda \varepsilon
\]
where $\varepsilon =(\varepsilon_{1},\varepsilon_{2},\ldots, \varepsilon_{n})^{t} \in \mathbb{R}^{n}$.  In other words $\lambda$ is an eigenvalue of the matrix $\mathcal{A}^{-1}\mathcal{B}$. Of course we cannot in general expect that $f\in \mathcal{V}$ if we make an arbitrary selection of test functions $f_{1},f_{2},\ldots, f_{n}$ but the $n$ eigenvalues of the symmetric matrix $\mathcal{A}^{-1}\mathcal{B}$ yield the Rayleigh-Ritz approximation to the lowest $n$ eigenvalues of the spectrum.  It can be demonstrated that, provided the set of test functions is complete, the eigenvalues of $\mathcal{A}^{-1}\mathcal{B}$ do `converge' to the spectrum of the Laplacian (see for example Chapter 11 of \cite{Strauss}). 

\subsection{Calculation of Relevant Integrals}\label{subsec: CRI}
We now apply the Rayleigh--Ritz approximation to the case of $(\mathbb{C},\Phi_{ps}^{\ast}(g_{Euc}))$.  The test functions we use are the pre-image under the standard stereographic map $\Phi_{1}:\mathbb{C}\rightarrow \mathbb{S}^{2}$ of the classical spherical harmonics which are the eigenfunctions of the Laplacian of the round metric $\Phi^{\ast}_{1}(g_{Euc})$.  If we define the functions
\[
X(z) = \frac{2\mathrm{Re}(z)}{1+|z|^{2}}, \qquad Y(z)=\frac{2\mathrm{Im}(z)}{1+|z|^{2}}, \qquad \mathrm{and} \qquad Z(z) = \frac{1-|z|^{2}}{1+|z|^{2}}, 
\]
then the first nine spherical harmonics can be written as
\[
f_{0}=1 \textrm{ (with eigenvalue 0 with respect to } \Phi_{1}^{\ast}(g_{Euc})),
\]
\[
f_{1}=X, \quad f_{2}=Y, \quad f_{3}=Z \textrm{ (all with eigenvalue 2)},
\]
\[
f_{4}=X^{2}-Y^{2}, \quad f_{5} = XY, \quad f_{6}=XZ, \quad f_{7}=YZ, \quad f_{8}=3Z^{2}-1,
\]
all with eigenvalue $6$. It is no difficulty to continue this process and write subsequent spherical harmonics in terms of polynomials in the functions $X,Y,$ and $Z$ (see for example \cite{KS}).\\
\\
The calculation of the matrix $\mathcal{B}$ is actually very straightforward as, for surfaces, the quantities $\langle \nabla f_{i}, \nabla f_{j}\rangle_{L^{2}(M)}$ are all conformally invariant (invariant under scalings of the metric $g \rightarrow e^{F}g$ for $F$ a function.).  The form of the metric (\ref{eqn:metric_form_gen}) $\Phi_{ps}^{\ast}(g_{Euc})$ makes it clear that the metric is conformally  equivalent to  $\Phi_{1}^{\ast}(g_{Euc})$ and so we can compute the integrals in  terms of this metric.  Integration-by-parts and classical formulae for integrating polynomials restricted to the sphere e.g.\cite{Fol}) yield the following.
\begin{proposition}
Let $(\mathbb{C},g=\Phi_{ps}^{\ast}(g_{Euc}))$, $f_{i}$ be the functions as defined previously and $\mathcal{B}$ defined by
\[
\mathcal{B}_{ij} = \iint_{\mathbb{C}}g(\nabla f_{i},\nabla f_{j}) dV_{g}.
\]
Then
\[
\mathcal{B}=\mathrm{Diag}\left(0,\frac{8\pi}{3},\frac{8\pi}{3},\frac{8\pi}{3},\frac{32\pi}{5},\frac{8\pi}{5},\frac{8\pi}{5},\frac{8\pi}{5}, \frac{96\pi}{5}\right)
\]
\end{proposition}
A more complicated prospect is the calculation of the integrals in the matrix $\mathcal{A}$.  From the form of the metric $\Phi_{ps}^{\ast}(g_{Euc})$ described in Equation (\ref{eqn:metric_form_gen}), and the form of the test functions $f_{i}$, we need to calculate integrals of the form
\begin{equation}\label{eqn:A_integral}
\mathcal{I} = \iint_{\mathbb{D}(a,r)}X^{p}Y^{q}Z^{s}\frac{2C}{(|x+iy-A|^{2}+B)^{2}}\ dx dy,
\end{equation}
where $\mathbb{D}(a,r) \subset \mathbb{C}$ is the closed disc of radius $r$ centred at $a\in \mathbb{C}$, $A\in \mathbb{C}$ and $B,C \in \mathbb{R}_{+}$. We claim that $\mathcal{I}$ has a closed form expressible in these variables (though the expression is extremely complicated). We describe the steps needed to find this expression.\\
\\
\textbf{Step 1: Rotate so that the domain is $\mathbb{D}(0,R)$, that is a disc centred at 0.}\\
\\
We observe that the form of the integrand (\ref{eqn:A_integral}) is preserved under the action of the group $PSU(2)$. Let 
\[
\gamma(z) = \frac{\alpha z+\beta}{-\overline{\beta}z+\alpha},
\]
with $\alpha \in \mathbb{R}$, $\beta \in \mathbb{C}$ and ${|\alpha|^{2}+|\beta|^{2}=1}$ be an element of $PSU(2)$ that maps a disc $\mathbb{D}(0,R)$ about the origin onto $\mathbb{D}(a,r)$ (one can always find such a map using Lemma \ref{lem:rot_lem}).  Straightforward algebra yields
\[
\mathcal{I} = \iint_{\mathbb{D}(0,R)}\gamma^{\ast}(X^{p}Y^{q}Z^{s})\frac{2\tilde{C}}{(|x+iy-\tilde{A}|^{2}+\tilde{B})^{2}} \ dxdy.
\]
Then to compute the terms  $\gamma^{\ast}(X^{p}Y^{q}Z^{s})$ that arise, it is useful to think of the spherical harmonics in the context of representations of the group $PSU(2)$. The following can be proved by tedious algebra and so we omit the proof.
\begin{lemma}
Let $\gamma \in PSU(2)$ be given by
\[
\gamma(z) = \dfrac{\alpha z+\beta}{-\overline{\beta}z+\alpha},
\]
where $\alpha \in \mathbb{R}$, $\beta \in \mathbb{C}$ and ${|\alpha|^{2}+|\beta|^{2}=1}$. Further, let $f_{i}$ be the spherical harmonic functions as given previously. Then
\[
\gamma^{\ast}(f_{0}) =f_{0}(\gamma(z)) = f_{0}. 
\]
Define $F$ in the span ${\langle f_{1},f_{2},f_{3} \rangle}$ by
\[
F(z) = \varepsilon_{1}f_{1}+\varepsilon_{2}f_{2}+\varepsilon_{3}f_{3},
\]
with $v_{F}:=(\varepsilon_{1},\varepsilon_{2},\varepsilon_{3}) \in \mathbb{R}^{3}$. Then
\[
\gamma^{\ast}(F)=F(\gamma(z)) = \mathcal{M}_{1}v_{F},
\]
where
\[
\mathcal{M}_{1} =
\begin{pmatrix}
\mathrm{Re}(\alpha^2-\bar{\beta}^2)&\mathrm{Im}(\alpha^2+\bar{\beta}^2)&-\mathrm{Re}(2\alpha \bar{\beta})\\
-\mathrm{Im}(\alpha^2-\bar{\beta}^2)&\mathrm{Re}(\alpha^2+\bar{\beta}^2)& 2\mathrm{Im}(\alpha \bar{\beta})\\
\alpha\beta+\overline{\alpha\beta}& 2\mathrm{Im}(\alpha\beta) &\alpha^2-|\beta|^2
\end{pmatrix}.
\]
Define $G$ in the span ${\langle f_{4},f_{5},f_{6},f_{7},f_{8} \rangle}$ by
\[
G(z) = \kappa_{4}f_{4}+\kappa_{5}f_{5}+\kappa_{6}f_{6}+\kappa_{7}f_{7}+\kappa_{8}f_{8},
\]
with $w_{G}:=(\kappa_{4},\kappa_{5},\kappa_{6},\kappa_{7},\kappa_{8})\in \mathbb{R}^{5}$. Then
\[
\gamma^{\ast}(G) = G(\gamma(z))=\mathcal{M}_{2}w_{G},
\]
where
\[
\mathcal{M}_{2}=\begin{pmatrix}
\mathrm{Re}(\alpha^4 + \beta^{4})  &    \mathrm{Re}(\beta^{2})\mathrm{Im}(\beta^{2})       & \alpha \mathrm{Re}(\beta)(2\mathrm{Re}(\beta^2)-1) &      \alpha \mathrm{Im}(\beta)(2\mathrm{Re}(\beta^2)+1) &  6\alpha^2\mathrm{Re}(\beta^{2}) \\

 4\mathrm{Re}(\beta^{2})\mathrm{Im}(\beta^{2}) & \mathrm{Re}(\alpha^4 - \beta^{4}) & -M_{32} & -M_{42} & 12\alpha^2\mathrm{Im}(\beta^{2})  \\
M_{31}& M_{32} &  M_{33} & \mathrm{Im}(\beta^{2}) (|\beta|^{2}-3\alpha^{2}) &                          12\mathrm{Re}(\beta)(\alpha- 2\alpha^3) \\
M_{41} &  M_{42} & \mathrm{Im}(\beta^{2}) (|\beta|^{2}-3\alpha^{2}) &  M_{44}  & 12\mathrm{Im}(\beta)(\alpha- 2\alpha^3)\\
2\alpha^2 \mathrm{Re}(\beta^{2}) &  \alpha^2 \mathrm{Im}(\beta^{2})& \alpha \mathrm{Re}(\beta) (\alpha^2- |\beta|^{2}) &  \alpha \mathrm{Im}(\beta) (\alpha^2- |\beta|^{2}) & 1-6|\alpha\beta|^2
\end{pmatrix},
\]

and where
\[
M_{31} = 4 \alpha \mathrm{Re}(\beta)(1-2\mathrm{Re}(\beta^{2})), \qquad
M_{41} = 4 \alpha \mathrm{Im}(\bar{\beta})(1+2\mathrm{Re}(\beta^{2})),
\]
\[
M_{32} = 2 \alpha \mathrm{Im}(\beta) (\alpha^2-2\mathrm{Re}(\beta^{2})-|\beta|^{2}), \qquad
M_{42} = 2 \alpha \mathrm{Re}(\beta)(\alpha^2 +2 \mathrm{Re}(\beta^{2})-|\beta|^{2}),
\]
\[
M_{33} = \alpha^{4}-6\alpha^{2}(\mathrm{Re}(\beta))^{2}+(\mathrm{Re}(\beta))^{4}-(\mathrm{Im}(\beta))^{4},
\quad \mathrm{and} \quad
M_{44} = \alpha^{4}-6\alpha^{2}(\mathrm{Im}(\beta))^{2}+(\mathrm{Im}(\beta))^{4}-(\mathrm{Re}(\beta))^{4}.
\]

\end{lemma}

The previous Lemma can be understood via representation theory of $PSU(2)$: the span of the constant function $f_{0}$ is the trivial representation; the functions $f_{1},f_{2},f_{3}$ span the 2-eigenspace of the Laplacian associated to the standard round metric $\Phi_{1}^{\ast}(g_{Euc})$ and are the three-dimensional (adjoint) representation;  the functions $f_{4},f_{5},\ldots,f_{8}$ span the $6$-eigenspace of the standard Laplacian and correspond to the irreducible five-dimensional representation of $PSU(2)$. It is clear how this process can be continued to the spherical harmonics associated to larger eigenvalues.\\
\\
\textbf{Step 2: Evaluation of integrals over $\mathbb{D}(0,R)$}\\
\\
We assume here that $A\neq 0$; otherwise the integrals can be calculated easily using standard formulae.
\\
To evaluate an integral of the form
\[
\mathcal{I} = \iint_{\mathbb{D}(0,R)}X^{p}Y^{q}Z^{s}\frac{2C}{(|x+iy-A|^{2}+B)^{2}} \ dx dy.
\]
we convert to polar coordinates where such an integral can be written as
\[
\mathcal{I} = \int_{0}^{R}\frac{r}{(1+r^{2})^{p+q+s}}\int_{0}^{2\pi} \frac{P(r\cos(\theta),r\sin(\theta))}{(r^{2}+|A|^{2}+B-2\mathrm{Re}(A)r\cos(\theta)-2\mathrm{Im}(A)r\sin(\theta))^{2}} d\theta dr
\]
for a two-variable polynomials $P$.  The integral 
\[
\Theta(r) := \int_{0}^{2\pi} \frac{P(r\cos(\theta),r\sin(\theta))}{(r^{2}+|A|^{2}+B-2\mathrm{Re}(A)r\cos(\theta)-2\mathrm{Im}(A)r\sin(\theta))^{2}} d\theta
\]
can be easily computed explicitly using the Residue Theorem method (for example outlined in Chapter 5 of \cite{Ahlfors}). We let $w=e^{i\theta}$ and convert to a contour integral. Hence
\[
\Theta(r) = -i\int_{|w|=1}\frac{P(\frac{r}{2}\left(w+\frac{1}{w} \right), \frac{r}{2i}\left(w-\frac{1}{w} \right))dw}{w(r^{2}+|A|^{2}+B-\mathrm{Re}(A)r\left(w+\frac{1}{w} \right)+i\mathrm{Im}(A)r\left(w-\frac{1}{w}\right))^{2}},
\]
and rearranging yields
\[
\Theta(r) = -i\int_{|w|=1}\frac{w P(\frac{r}{2}\left(w+\frac{1}{w} \right), \frac{r}{2i}\left(w-\frac{1}{w} \right))dw}{((r^{2}+|A|^{2}+B)w-r\overline{A}w^2-rA)^{2}}.
\]
The poles inside the unit circle $|w|=1$ are at the points
\[
w= \dfrac{-(r^{2}+|A|^{2}+B)+\sqrt{(r^{2}+|A|^{2}+B)^{2}-4|A|^{2}r^{2}}}{2(-r\overline{A})},
\]
and, depending upon the precise form of $P$, $w=0$ (it is straightforward to check that the other root of the quadratic denominator lies outside the unit circle). 
The resulting explicit expression in $r$ can then be used in the integral
\[
\mathcal{I} = \int_{0}^{R}\frac{r}{(1+r^{2})^{p+q+s}}\Theta(r)dr.
\]
we have used Mathematica to do this up to all the quartic integrals needed to compute the matrix $\mathcal{A}$ for the nine test functions \cite{wolfram}. The code and data used to produce the results presented in this paper are available on GitHub (\url{https://github.com/RPirie96/RGMolSA}).\\
\\
The calculation of these functions is one-off; once uploaded into the routine, users will just call them when the value of the integral is needed. If we wished to approximate eigenvalues beyond the first nine, it might not be practical to produce such closed-form expressions for these integrals.  In this case, one could use the residue method to produce the radial function $\Theta(r)$ (the expression for this function is relatively simple) and then use a numerical method such as Romberg integration to compute an approximate value.\\
\\
As mentioned already, one deficiency in the method is that for a given choice of base sphere, higher-level spheres are parameterised by ever smaller discs in the complex plane; this renders them higher frequency objects only detectable by eigenvalues further up the spectrum. Additionally, the evaluation of the explicit functions for extremely small values of the input parameters can be dominated by numerical errors. We stabilise our calculation by discounting the contribution of any sphere where any input has numerical value smaller  than $10^{-9}$.  

{ \subsection{Computing the Similarity Between Two Descriptors} \label{section:similarity}

There are many possible methods to compare two vectors. The convention in chemical similarity searching is to choose a measure bounded by 0 and 1, where 0 means two descriptors share no common features, and 1 means they are identical. This choice allows for easy interpretation of similarity scores. In this paper we use the Bray-Curtis distance, originally introduced to compare the presence of a species across two sites, with uses in botany, ecology and environmental science~\cite{Bray_Curtis_1957}. It has subsequently been adopted across a wide range of scientific disciplines, and is useful for our purpose as it is naturally bound by 0 (identical) and 1 (no similarity). We make use of the Bray-Curtis distance implemented in SciPy~\cite{scipy}. For our purpose we define the similarity between two vector descriptors $u$ and $v$ as the inverse Bray-Curtis distance to fit with the previous convention of a similarity score of 1 when two molecules are identical:
\[
d(u,v) = 1 - \frac{\sum |u_{i}-v_{i}|}{\sum |u_{i}+v_{i}|}.
\]
\\
}

\section{Initial Case Study: Phosphodiesterase 5 (PDE5) Inhibitors}
\label{sec:4}
\subsection{A Worked Example of Generating the Shape Descriptor: Sildenafil}
As an illustration of the theory described in the preceding sections, we consider the drug Sildenafil. We describe the space filling (or CPK) model of the molecule by obtaining the atomic centres, adjacency matrix and van der Waals radii for a 3D embedding of the molecule (generated using the ETKDG algorithm \cite{Riniker_Landrum_2015}) from the RDKit cheminformatics package~\cite{landrum}. As shown in Figure~\ref{fig:Sildenafil_remove_ring}, we then replace each of the four rings in the molecule with a single sphere using the process outlined in Section \ref{section:moldata}. We compute the centroid, find the level-0 sphere, which is one of the spheres that replaced a ring, and then label the remaining spheres with the correct level (Figure~\ref{fig:Sildenafil_levels}).
\begin{figure}
\centering\includegraphics[width=5in]{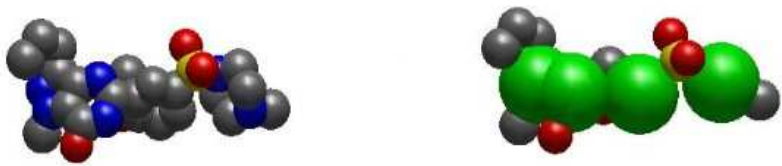}
\caption{(left) The standard CPK representation of the surface associated to Sildenafil (with hydrogen atoms removed). (right) The surface after the rings have been replaced by spheres of the same radius.} \label{fig:Sildenafil_remove_ring}
\end{figure}
\begin{figure}
\centering\includegraphics[width=3in]{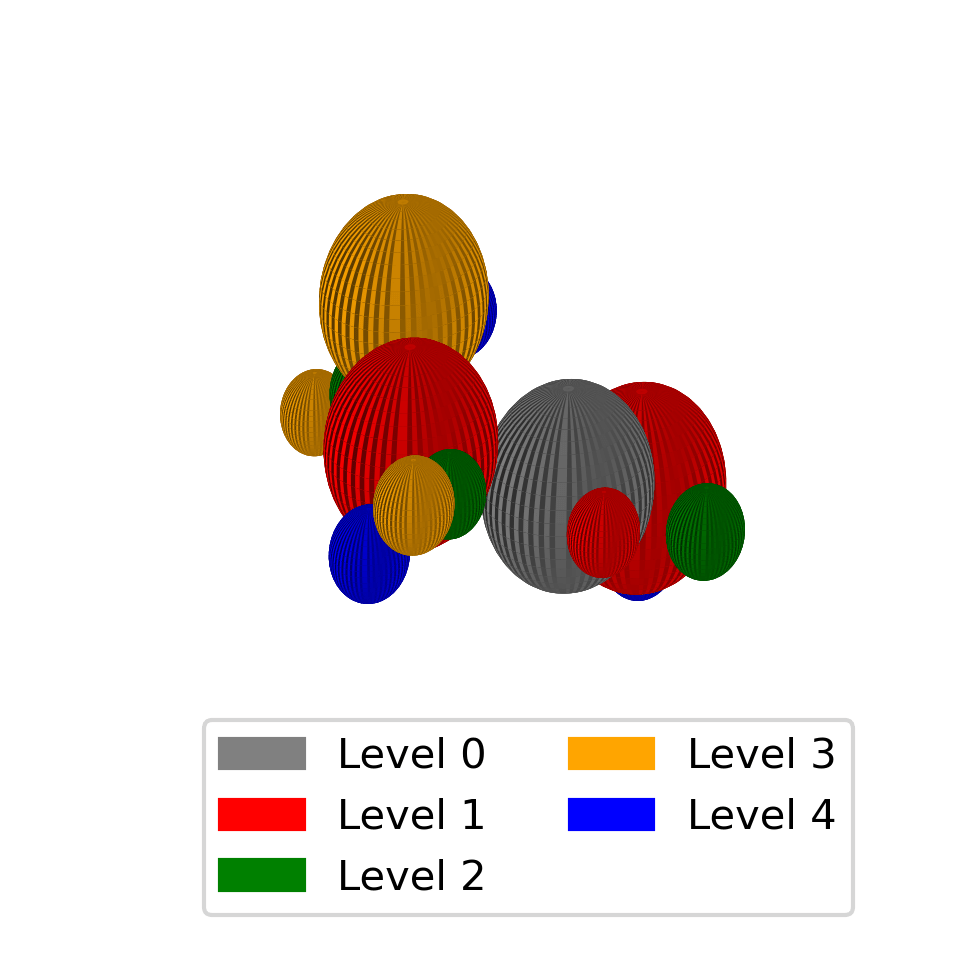}
\caption{Each sphere has a level determined by how many spheres connect it to the grey level-0 sphere} \label{fig:Sildenafil_levels}
\end{figure}
\noindent
We find that the $16$ spheres are distributed across $5$ levels (from $0$ to $4$). The surface area is re-scaled (from the original value of $314.57$~\AA{}$^2$) to $4\pi$, and piecewise stereographic projection is performed.  The resulting centres and radii of the discs in the complex plane are summarised in Table \ref{tab:Sildenafil_data}. Table \ref{tab:Sildenafil_data} shows that the higher-level spheres get `crunched' into discs of very small radii.  For large molecules, such domains will be numerically indistinguishable from being empty. This is a defect caused by using just a single coordinate chart to describe the surface; in the sequel we will address this issue.  It might also be possible to form a shape descriptor by judicious selection of a variety of base spheres and computing the spectrum attached to each of these (see Section \ref{section:conclusion}).

It is possibly more illuminating to view a plot of the discs which are the domains used to construct the piecewise stereographic projection map $\Phi_{ps}$ (and the domains of definition of the metric in Equation \ref{eqn:met_form}), as shown in Figure \ref{fig:Sildenafil_disc}. As we will see in the next section the resulting non-zero eigenvalues of $\mathcal{A}^{-1}\mathcal{B}$ can then be used to construct the shape descriptor, $v_{\mathcal{S}}$.

\begin{table}
\caption{Data describing the induced metric on the molecular surface of Sildenafil. `Sphere' is the numbering of the sphere using data from RDKit and after removing rings; `Level' records the level of the sphere; $A$, $B$ and $C$ are the quantities used to describe the metric in Equation (\ref{eqn:metric_form_gen}).}
\centering
\begin{tabular}{|c|c|c|c|c|c|c|c|}
\hline
 Sphere &  Level &        Centre &  Radius &             A &            B &            C \\
\hline
     1 &      0 &  0.0+0.0j &   $\infty$ & -0.000+0.000j & 1.000 & 0.404 \\
     \hline
      0 &      1 &  0.827+1.129j &   1.272 &  0.375+0.512j & 0.299 & 0.121 \\
      2 &      1 & -1.221-0.487j &   0.492 & -1.112-0.444j & 0.008 & 0.003 \\
      7 &      1 &  0.872-0.676j &   0.378 &  0.803-0.622j & 0.037 & 0.003 \\
      \hline
      6 &      2 &  0.084+0.816j &   0.167 &  0.105+0.794j & 0.008 & $6.3\times 10^{-4}$ \\
      8 &      2 & -1.127-0.395j &   0.013 & -1.127-0.396j & $8.3\times 10^{-6}$ &  $7.7\times 10^{-7}$\\
     12 &      2 & -1.062-0.613j &   0.062 & -1.067-0.595j &  $5.6\times 10^{-4}$ & $3.9\times 10^{-5}$\\
     15 &      2 &  0.784+0.321j &   0.182 &  0.750+0.337j & 0.010 &  $8.4\times 10^{-4}$\\
     \hline
      3 &      3 & -1.129-0.395j &   0.003 & -1.128-0.396j & $2.9\times 10^{-7}$ & $1.2\times 10^{-7}$\\
      5 &      3 &  0.071+0.734j &   0.051 &  0.078+0.747j & $3.3\times 10^{-4}$ & $2.7\times 10^{-5}$ \\
      9 &      3 & -1.125-0.398j &   0.002 & -1.126-0.398j & $3.8\times 10^{-7}$ & $2.7\times 10^{-8}$ \\
     10 &      3 & -1.126-0.394j &   0.002 & -1.126-0.394j & $3.6\times 10^{-7}$ & $2.5\times 10^{-8}$ \\
     13 &      3 & -1.072-0.583j &   0.012 & -1.071-0.585j & $1.9\times 10^{-5}$ & $1.6\times 10^{-6}$ \\
     \hline
      4 &      4 &  0.088+0.733j &   0.012 &  0.086+0.736j & $1.6\times 10^{-5}$ & $1.4\times 10^{-6}$ \\
     11 &      4 & -1.128-0.395j &   0.000 & -1.128-0.395j & $3.2\times 10^{-9}$ & $2.6\times 10^{-10}$ \\
     14 &      4 & -1.069-0.582j &   0.003 & -1.070-0.583j & $9.9\times 10^{-7}$ & $8.2\times 10^{-8}$ \\
\hline
\end{tabular}
\label{tab:Sildenafil_data}
\end{table}

\begin{figure}
\centering\includegraphics[width=3.75in]{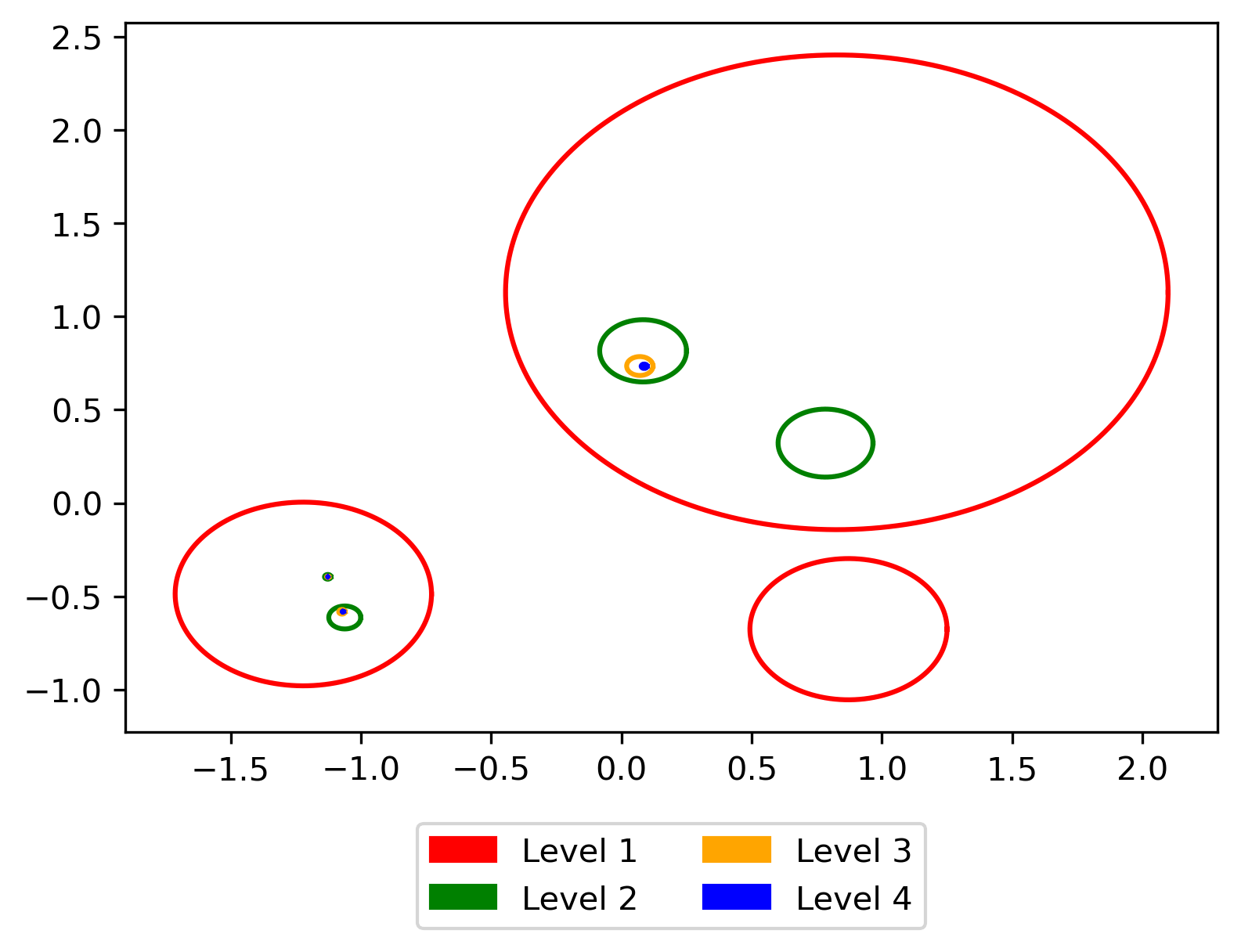}
\caption{The `discs within discs' picture associated to the domains of piecewise stereographic projection for Sildenafil. } \label{fig:Sildenafil_disc}
\end{figure}

{\subsection{Weighting the Surface Area} \label{section:SAweight}

As discussed in Section \ref{section:moldata}, the surface of the molecule must be re-scaled to have an overall area of $4\pi$. We account for this re-scaling by trialling a few methods to include the original un-scaled surface area ($A$, measured in \AA$^{2}$) in our final descriptor.\\
\\
The shape descriptor for Sildenafil, formed by taking the area and the non-zero eigenvalues of $\mathcal{A}^{-1}\mathcal{B}$, and using the un-scaled surface area as the first element of the vector, is given by:
\[
v_{\mathcal{S}} = (314.57, 0.98, 4.08, 8.00, 8.39, 15.64, 19.45, 20.91, 26.64).
\]
As discussed in the Introduction, we would expect high similarity between Sidenafil and Vardenafil (a follow-up drug with very similar chemical structure) and reasonable similarity to Tadalafil (an inhibitor targeting the same pocket, but with a different chemical scaffold). This is indeed what we see in the first column of Table~\ref{tab:SAweight} (labelled $A$), using the un-scaled surface area in the descriptor, with all similarity measures higher than 0.8 and the highest similarity (0.955) between Sidenafil and Vardenafil.\\
\\
However, we were concerned that the surface area may dominate to the similarity score (meaning that only molecules of similar size would be considered similar), and so we also considered replacing the first element of the vector by the inverse area, weighted by a scaling factor ($a$):
\begin{align}
    \lambda_{0} = \frac{a}{A}
\end{align}
When $a=1$, the surface area term contributes negligibly to the similarity measure. Although the similarities between all three molecules drop off, gratifyingly, there is no change in the relative ordering. This pattern continues as $a$ is increased, and in what follows we use $a = 10^4$ as a compromise, which gives $a/A$ of order 1, comparable to the larger eigenvalues in the surface descriptor.  }

\begin{table}
\centering
\caption{The inverse Bray-Curtis distance between a series of PDE5 inhibitors. The first element of the vector descriptor is replaced by either the un-scaled surface area ($A$), or the inverse area (scaled by $a$).}
\begin{tabular}{|c|c|c|c|c|c|c|} 
\hline 
& A & $a = 10^0$ & $a = 10^1$ & $a = 10^2$ & $a = 10^3$ & $a = 10^4$ \\
\hline
Sildenafil-Vardenafil & 0.955 & 0.882 & 0.882 &  0.882 & 0.885 & 0.903 \\
\hline
Sildenafil-Tadalafil & 0.875 & 0.768 & 0.768 &  0.769 & 0.774 & 0.809 \\
\hline
Vardenafil-Tadalafil & 0.833 & 0.667 & 0.667 & 0.668 & 0.675 & 0.725 \\
\hline
\end{tabular}
\label{tab:SAweight}
\end{table}

\subsection{Investigating Variation in 3D Conformers}
Molecules, on the whole, are not rigid entities and can adopt different orientations known as conformers. In order to bind to the protein of interest, and by extension act as a drug, a molecule adopts a conformer that best complements the shape of the binding pocket. It is therefore important to consider how the shape approximation given by our method relates conformers of the same molecule. In theory, these should be more self-similar than they are similar to different molecules. In chemoinformatics, two molecules with a similarity score of $0.7$ or above are typically considered similar.

Here, we consider two small sets of 10 conformers for each our PDE5 inhibitor examples: one set of low energy conformers, which we would expect to have higher similarity, and one set of random conformers, for which we would expect slightly more variance. This is by no means an exhaustive sample, but gives an indication of the general trend. We produced these sets using the ETKDG algorithm \cite{Riniker_Landrum_2015} with energy optimisation using the MMFF94 force field~\cite{tosco_stiefl_landrum_2014}, both implemented in RDKit~\cite{landrum}.
 
Figure \ref{fig:conf_sim} shows the minimum, maximum and average shape similarity, as well as the average root-mean-square deviation (RMSD) for each set. The RMSD is widely used to compare molecular conformers based on their atomic positions, rather than molecular surface. The full set of RMSD and shape similarity comparisons are available in the {\bf Supporting Data}. 

In general these follow the expected trend of high self-similarity, with the exception of the Sildenafil random conformer set, where a slightly lower shape similarity score is obtained. However this does still fall above the standard threshold of 0.7, allowing these to still be classified as similar. This small initial sample indicates that our method shows promising potential in handling molecular conformations, however this will ultimately need verification with a larger sample size.

\begin{figure}
\centering
\includegraphics[width=\textwidth]{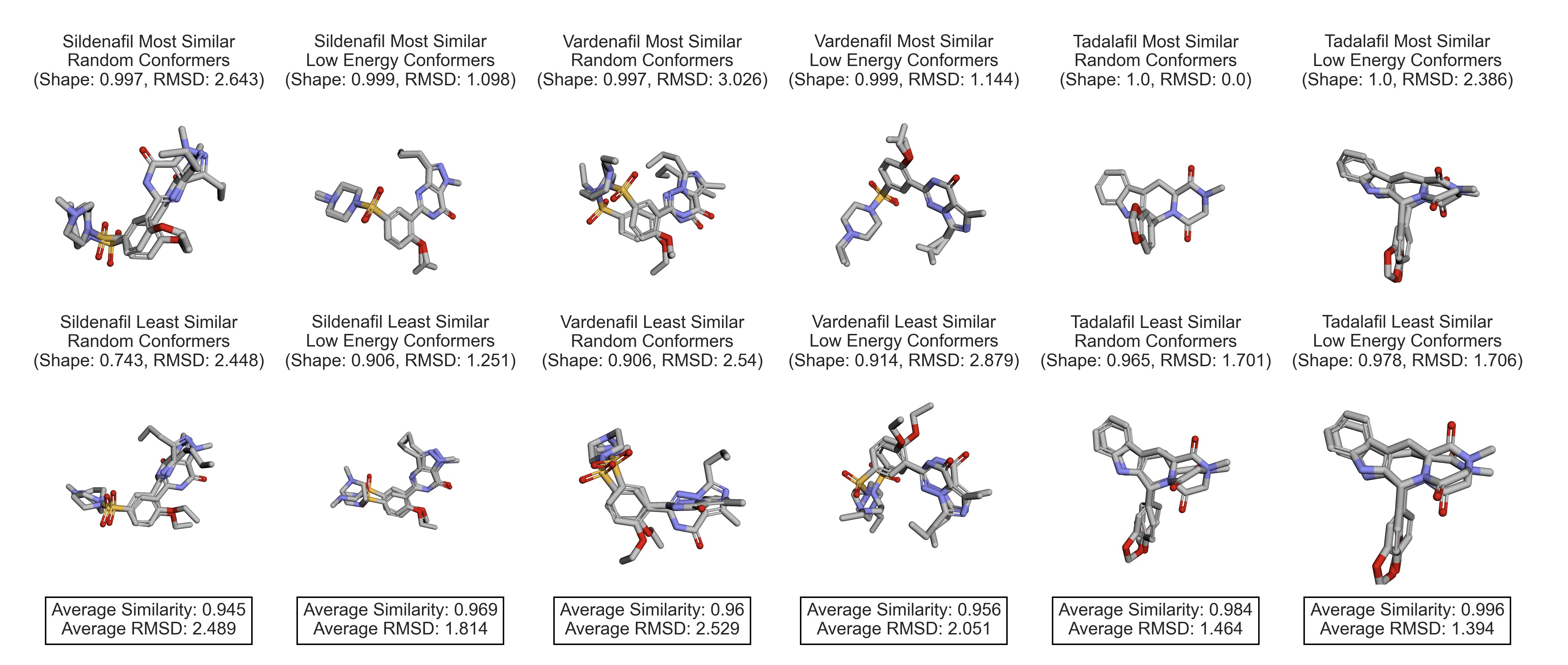}
\caption{ Overlay of the most and least shape-similar conformers of Sildenafil, Vardenafil and Tadalafil and the average shape similarity and RMSD for each set. On average the conformers display a high degree of self-similarity despite the variance in atom-position similarity} \label{fig:conf_sim}
\end{figure}

The swarm plots in Figure~\ref{fig:simvsrms} compare the RMSD and shape similarity scores for each set of comparisons. In general, across all 10 conformers in each set the shape similarity remains consistently high despite the variance in RMSD similarity. This shows that even conformers considered genuinely different are classified as having similar shape, which is consistent with the expectation that the spectrum of the Laplace–Beltrami operator is insensitive to surface deformation. This should have advantages for virtual screening in drug discovery, since i) no pre-alignment step of the molecules is needed and ii) molecules that can potentially deform to fit the binding pocket may be identified as potential hits, even if atomic coordinates based approaches classify them as the wrong shape. 

\begin{figure}%
    \centering
    \subfloat[][\centering RMS Similarity]{\includegraphics[width=7.5cm]{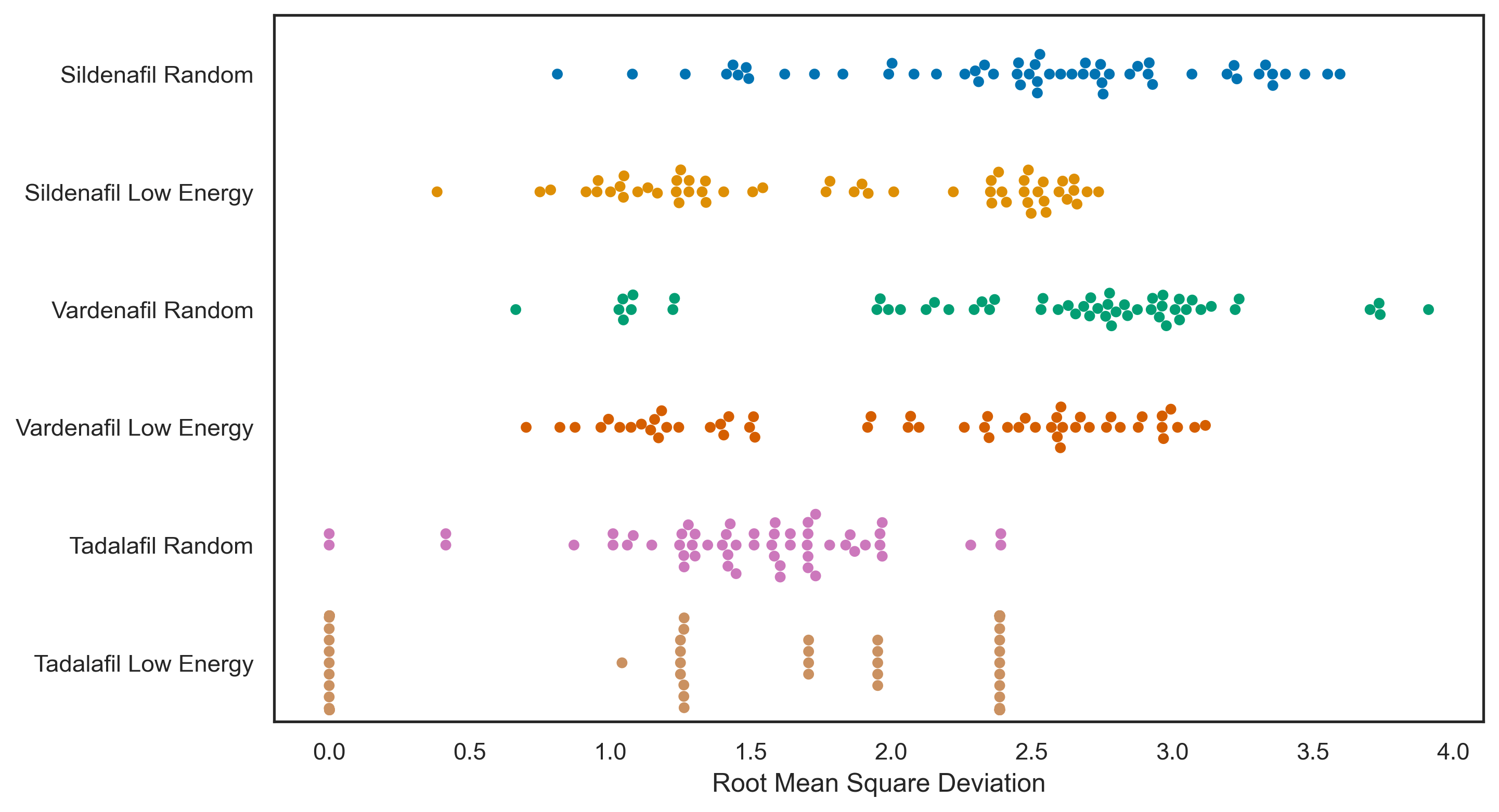} }%
    \qquad
    \subfloat[][\centering Shape Similarity]{\includegraphics[width=7.5cm]{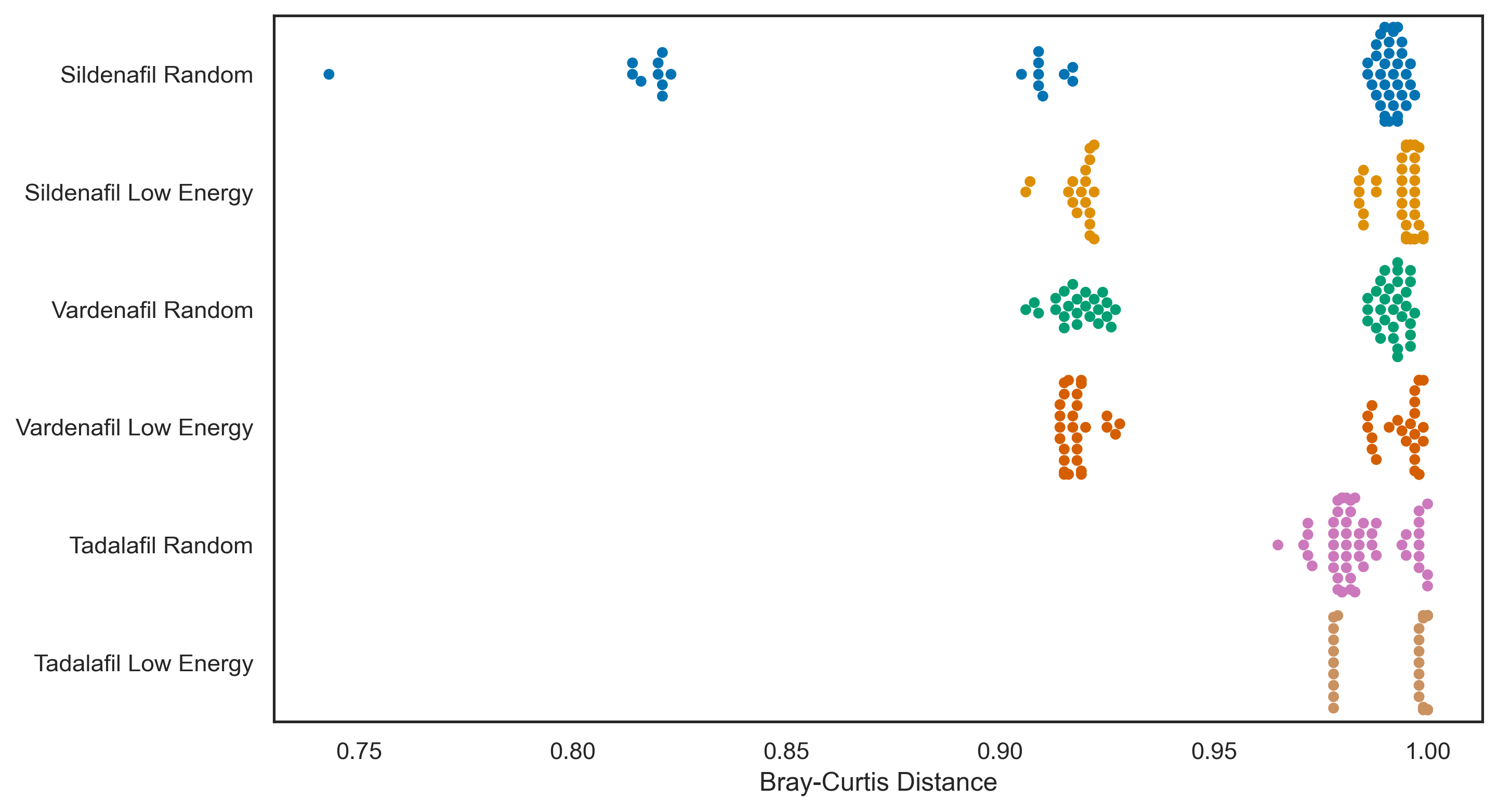}}%
    \caption{Swarm plots of the RMSD (in \AA) and shape similarity for our set of conformers highlight the general trend that different conformers are classed as having similar shape, despite significant variance in their atomic positions. Conformers with RMSD less than 1~\AA{} are considered similar, while those over 3~\AA{} have significant differences. 
    }%
    \label{fig:simvsrms}%
\end{figure}

\subsection{Comparison to Existing Methods}

To gain initial insight into how well our method might compare to the existing software in the field, we carried out a small study on the PDE5 inhibitors shown in Figure \ref{fig:scafhop}. Table~\ref{tab:SVTcomparison} compares our method to the atomic-distance based USRCAT descriptors implemented in RDKit~\cite{Schreyer_Blundell_2012, landrum}, Shape-It (an open-source version of ROCS)~\cite{Taminau_Thijs_De_Winter_2008}, and MolSG, the molecular surface descriptor presented by Seddon \textit{et al.}~\cite{SCPG}. We also include comparison to the 1024-bit Morgan fingerprint, using radius 3, which is a 2D representation of molecules in binary format, implemented in RDKit~\cite{landrum}. In all cases a similarity score is given between 0 (different) and 1 (identical). 

As discussed, we would expect that Sildenafil and Vardenafil should have a similarity score close to 1 for both the shape- and fingerprint-based methods as they are close structural analogues, where only a few small modifications have been made to the same core structure. Given it is known to have a similar volume to Sildenafil, we would anticipate Tadalafil also to score highly by shape, but to be less similar than Sildenafil and Vardenafil. We would expect the 2D method to classify these as different. Here one conformer is considered for each molecule for simplicity. 

Surprisingly USRCAT and Shape-It significantly underestimate the similarity of the Sildenafil analogues compared to our expectations (similarity $<0.4$, Table~\ref{tab:SVTcomparison}). The similarity of Sildenafil and Vardenafil is reasonably high in the Morgan fingerprint approach, but Tadalafil would perhaps not be identified as similar. Using MolSG, Tadalafil is actually identified as more similar to the two closely related analogues, Sildenafil and Vardenafil, than they are to each other. The RGMolSA approach discussed here gives the highest Sildenafil--Vardenafil similarity of the five, and recapitulates the expected ranking, suggesting that it is a promising approach for full-scale similarity searching in drug discovery.
\begin{center}
\begin{table}
\caption{Comparison of our method to existing atomic-distance \cite{Schreyer_Blundell_2012}, atomic-centred \cite{Taminau_Thijs_De_Winter_2008} and molecular surface based \cite{SCPG} descriptors. In all cases the similarity scores given are bound by 0 (no similarity) and 1 (identical).}
\begin{tabular}{|c|c|c|c|c|c|} 
\hline 
 & {\bf RGMolSA (this work)} & {\bf USRCAT} & {\bf Shape-It} & {\bf MolSG} & {\bf Morgan Fingerprint}\\
\hline
Sildenafil-Vardenafil & 0.903 & 0.384 & 0.388 &  0.704 & 0.667 \\
\hline
Sildenafil-Tadalafil & 0.809 & 0.269 & 0.278 &  0.746 & 0.201 \\
\hline
Vardenafil-Tadalafil & 0.725 & 0.291 & 0.353 & 0.887  & 0.209 \\
\hline
\end{tabular}
\label{tab:SVTcomparison}
\end{table}
\end{center}

\subsection{Similarity to Potential Decoys}

\begin{figure}
\centering
\includegraphics[width=4in]{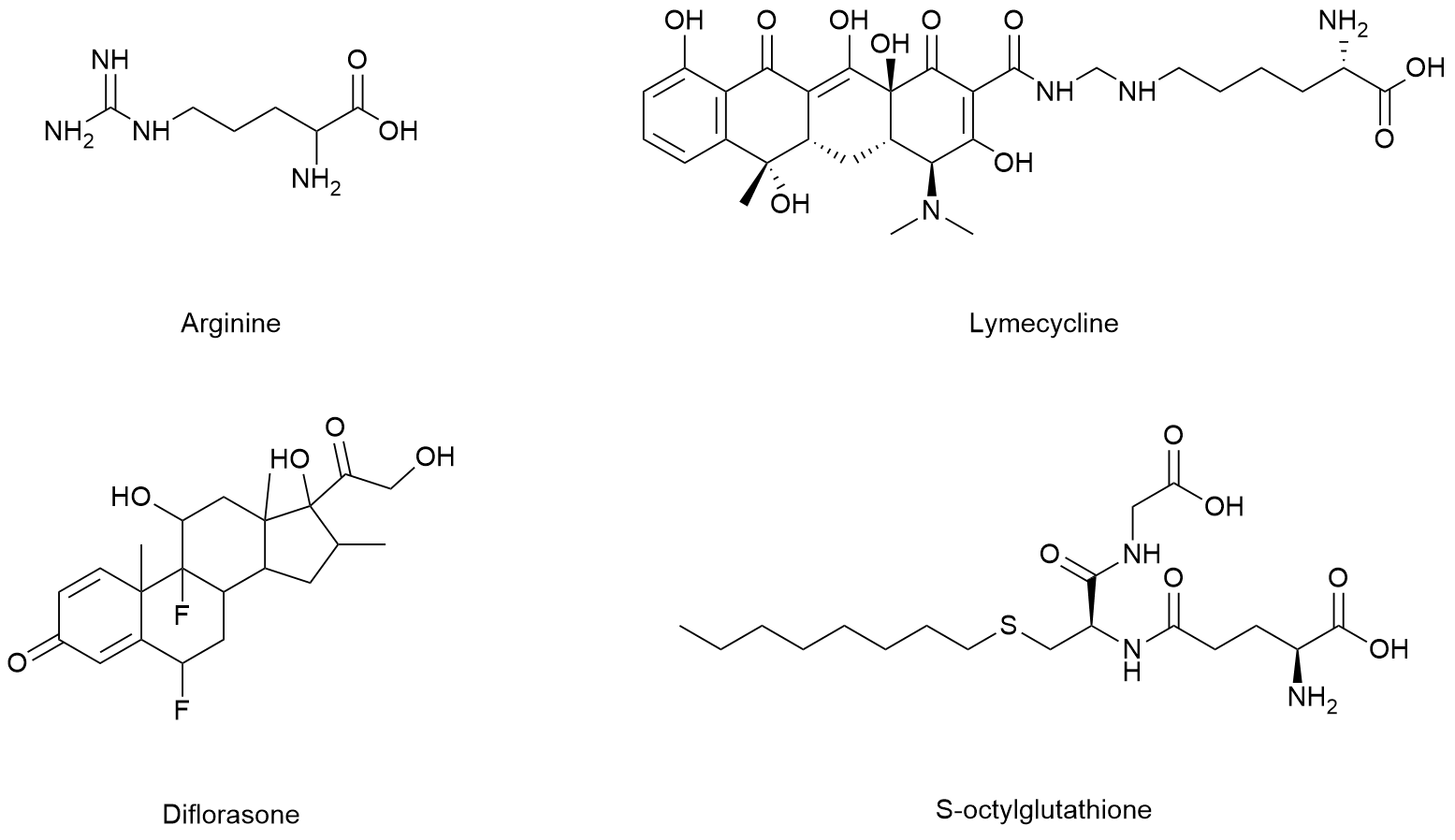}
\caption{Control case molecules } \label{fig:others}
\end{figure}
%
{The above comparison of PDE5 inhibitor molecules using the RGMolSA method gives generally high similarity scores. To demonstrate that our method does not simply classify all molecules as having similar shape, we present an additional comparison of these to four other molecules (Figure~\ref{fig:others}). Arginine was selected due to its lower molecular weight, and therefore smaller size, but same general shape (a long chain of spheres). Lymecycline has a higher molecular weight, but the four-ring motif potentially gives part of the molecule a similar shape to Sildenafil (especially if one of the ring atoms is used as the base sphere). Diflorasone has a similar molecular weight and four rings, but has a different therapeutic target/indication. S-octylglutathione again has similar molecular weight, but no rings and the potential for similarity due to the branching in the centre of the molecule. 
}

{ The results of this comparison are presented in Figure~\ref{fig:othersoverlay}. Generally the similarity scores range between 0.3--0.7, so these would unlikely be identified as potential hits in a virtual screen. Only three out of the twelve comparisons would be classed as `similar' (i.e. a similarity score $\geq 0.7$): Vardenafil with Arginine, Sildenafil with Diflorasone and Tadalafil with Diflorasone. The general shape of Sildenafil, Tadalafil and Diflorasone is similar: all contain 4 rings and have similar surface area. High similarity between Tadalafil and Diflorasone in particular would be expected as both have 4 fused rings in their structure. Looking at the structures of Vardenafil and Arginine, despite one containing rings while the other does not, the general shape of chain of spheres (whether those are rings or atoms) is common between the two. This suggests the potential capability for scaffold hopping within our method. 
}
\begin{figure}
\centering
\includegraphics[width=6.5in]{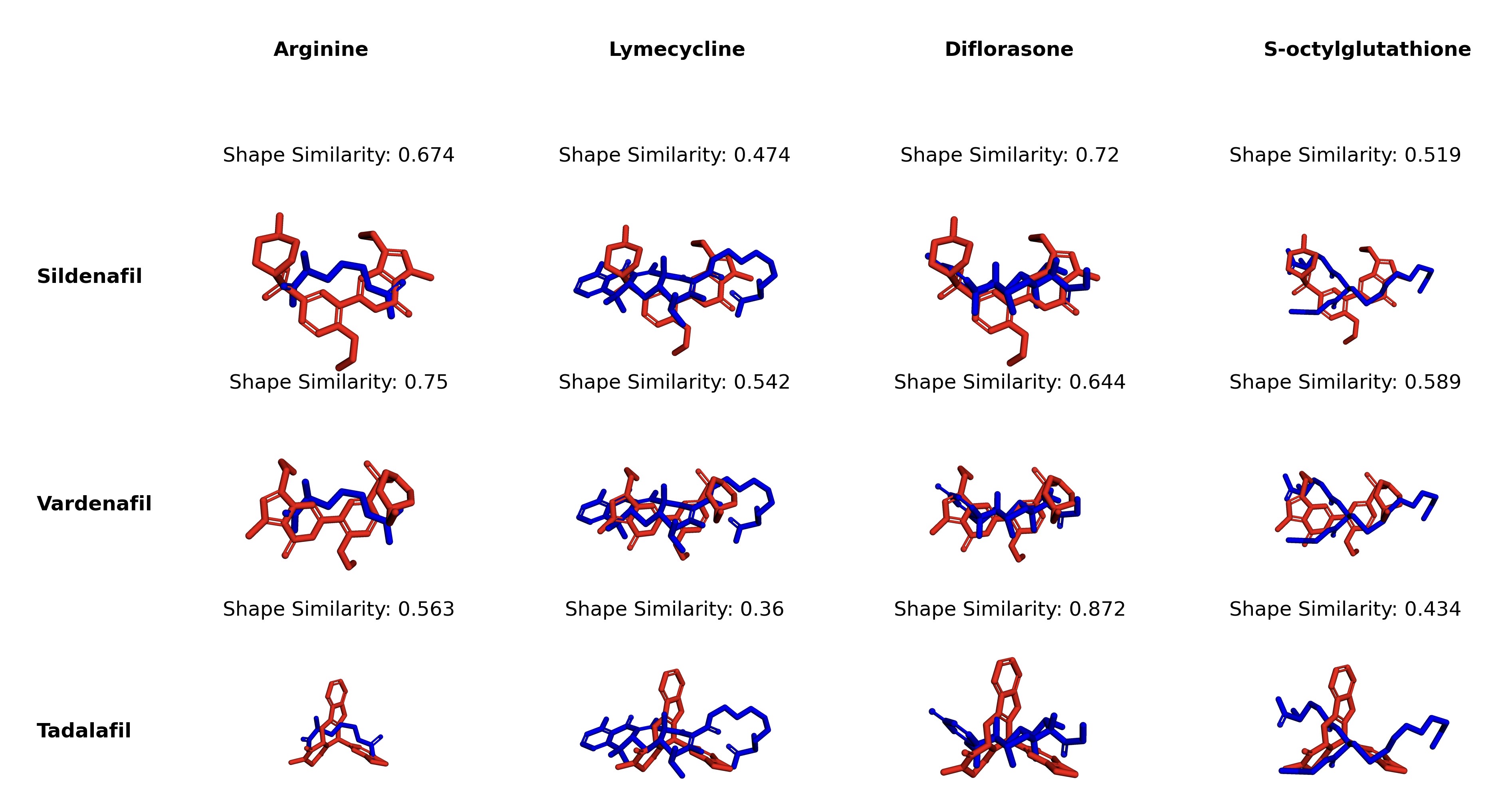}
\caption{Inverse Bray-Curtis similarity of four `different' molecules (blue) to the PDE5 inhibitor test series (red). The overlay of the structures was computed using Open3DAlign \cite{tosco_balle_shiri_2011}}. \label{fig:othersoverlay}
\end{figure}

\section{Conclusion}\label{section:conclusion}
We have outlined the theory underpinning a new molecular shape descriptor, which approximates the surface with a simple vector
\[
v_{\mathcal{S}} =(10^4/A,\lambda_{1},\lambda_{2},\ldots, \lambda_{8})\in \mathbb{R}^{9}.
\]
where $A$ is the surface area of the molecule (in \AA{}$^2$).
The descriptor is derived from an an explicit description of the Riemannian metric $g$ associated with the molecular surface $\mathcal{S}$ (after replacement of rings with spheres) constructed by considering the atomic centres, radii, and the adjacency matrix associated with the molecule. The form of the metric is given in Theorem 2.1 where we have used piecewise stereographic projection to produce a `discs within discs' picture of the surface (Figure \ref{fig:Sildenafil_disc}). From this concrete description we have demonstrated that it is possible to compute a Rayleigh-Ritz approximation to the Laplacian explicitly and without producing a mesh to approximate the surface using $v_{\mathcal{S}}$. The similarity between two descriptors is computed using the inverse Bray-Curtis distance, giving a score bound by 0 (completely different) and 1 (identical).\\
\\
The capabilities of our method were investigated using a series of PDE5 inhibitors known to have similar shape to each other - Sildenafil, Vardenafil and Tadalafil. The similarity between conformers of the same molecule is generally handled well, with scores above 0.8 in most cases. This matches the expectation that different conformers of the same molecule will generally be more self-similar than they are similar to other molecules. An initial comparison to other existing 3D shape similarity methods and a 2D based molecular fingerprint revealed our approach to be promising at quantifying similarity, outperforming all of the existing methods. A full retrospective benchmarking study will be required to verify this capability. An additional comparison of the PDE5 inhibitor examples to a set of potential decoys, which in theory would not be active against PDE5, was completed to prove that not all molecules are classed as similar. In most cases these had scores between 0.3--0.7, and would be classed as inactive. Vardenafil and Arginine, Sildenafil and Diflorasone and Tadalafil and Diflorasone each gave scores above 0.7, with the similarity between Vardenafil and Arginine in particular evidencing potential for scaffold hopping within our new method.

As a shape descriptor, $v_{\mathcal{S}}$ (and the method used to produce it) has some clear problems. While the descriptor $v_{\mathcal{S}}$ is alignment-free in the sense it is invariant under uniform rotations and translations of the initial molecular data, it depends upon a choice of base sphere (which determines the test functions for the Rayleigh-Ritz method). For large molecules, this has the effect of reducing the contribution to the approximation of atoms that are far from the base (in the sense of a path in the 2-D graph of the molecule). Even worse, if the molecule is large enough, the disc within disc data becomes so small as to cause numerical errors to accumulate in the approximation. For now these errors are handled by ignoring any contributions from regions with radii less than $10^{-9}$, however there are other more robust remedies to this problem:
\begin{enumerate}
    \item RGMolSA only computes the spectrum for a single base sphere, nearest to the centroid. We could pick other natural atoms in the molecule from which to compute the spectrum (similar to the approach taken by USRCAT \cite{Schreyer_Blundell_2012}), then combine these to give a better overall description of the surface.
    \item We could compute a more eigenvalues; this might not be possible to do in a closed form due to the complexity of the computer algebra needed to produce the formulae, but it will be possible to do quickly and accurately using the Romberg method to approximate the radial integral.
    \item We could pick a different set of test functions; it might be natural to use, in the notation of Subsection \ref{subsec: CRI}, $X(z),Y(z),Z(z)$ where we use the relevant elements of $PSL(2,\mathbb{C})$ to produce a coordinate $z$ for every single sphere in the surface.  The corresponding integrals would be slightly more complicated to evaluate but a very similar method should work (e.g. residue calculus). The resulting descriptor would not be biased in favour of any particular base sphere.
\end{enumerate}

The shortcomings of RGMolSA will also be addressed in our subsequent work using the Riemannian metric in the theory of K\"ahler quantisation to produce a novel shape descriptor, which approximates the surface shape at a global level. Both descriptors could also be further improved by introducing consideration of pharmacophoric features (e.g. aromatic rings, hydrogen bond donors and acceptors). For example, our method identifies Vardenafil and Arginine as having high shape similarity; based on this alone Arginine would be selected as a potential PDE5 inhibitor. If these two molecules have poor pharmacophoric overlap, Arginine may not actually bind, so consideration of these features will allow for more rounded predictions of activity.

\section{Acknowledgements}
The authors acknowledge support from an EPSRC Doctoral Training Partnership studentship (grant EP/R51309X/1), the Alan Turing Institute Enrichment Scheme (R.P.), and a UKRI Future Leaders Fellowship (grant MR/T019654/1) (D.J.C.). S.J.H. would like to thank Prof. T. Murphy and Dr R. L. Hall for their interest and for useful conversations about the project. We thank Dr A. Asaad for useful comments on a draft of the article.

\bibliographystyle{unsrt} 
\bibliography{references}

\end{document}